\documentclass[article]{IEEEtran}
\usepackage{amsmath}
\usepackage{amsfonts}
\usepackage{amssymb}
\usepackage{amsthm}
\usepackage{tabularx}
\usepackage{mathrsfs} 

\usepackage{url}
\usepackage[colorlinks]{hyperref}

\newtheorem{remark}{Remark}

\newtheorem{proposition}{Proposition}
\newtheorem{lemma}{Lemma}
\newtheorem{corollary}{Corollary}

\usepackage{stfloats}
\usepackage{float}
\usepackage{graphicx}
\usepackage{xcolor}
\usepackage{subfigure}
\usepackage{array}
\makeatletter
\newcommand{\thickhline}{%
	\noalign {\ifnum 0=`}\fi \hrule height 1pt
	\futurelet \reserved@a \@xhline
}
\newcolumntype{"}{@{\hskip\tabcolsep\vrule width 1pt\hskip\tabcolsep}}
\makeatother

\makeatletter
\def\blfootnote{\xdef\@thefnmark{}\@footnotetext}
\makeatother

\begin{document}

\title{\huge Enormous Fluid Antenna Systems (E-FAS) under Correlated Surface-Wave Leakage: Physical Layer Security}

\author{Farshad~Rostami~Ghadi,~\IEEEmembership{Member},~\textit{IEEE}, 
Kai-Kit~Wong,~\IEEEmembership{Fellow},~\textit{IEEE}, 
Masoud~Kaveh,~\IEEEmembership{Member}, \textit{IEEE}, \\
Mohammad~Javad~Ahmadi, 
Kin-Fai~Tong,~\IEEEmembership{Fellow},~\textit{IEEE}, and 
Hyundong Shin, \IEEEmembership{Fellow},~\textit{IEEE}

\vspace{-8mm}
}

\maketitle

\blfootnote{The work of  F. Rostami Ghadi and K. K. Wong is supported by the Engineering and Physical Sciences Research Council (EPSRC) under Grant EP/W026813/1.}
\blfootnote{The work of M. J. Ahmadi was supported by the German Research Foundation (DFG) as part of Germany's Excellence Strategy -- EXC 2050/1 -- Project ID 390696704 -- Cluster of Excellence ``Centre for Tactile Internet with Human-in-the-Loop'' (CeTI) of Technische Universit\"{a}t Dresden.}
\blfootnote{The work of H. Shin is supported by the National Research Foundation of Korea (NRF) grant (RS-2025-00556064 and RS-2025-25442355), funded by the Ministry of Science and ICT (MSIT), Korea, under the ITRC (Information Technology Research Center) support program (IITP-2025-RS-2021-II212046), supervised by the IITP (Institute for Information \& Communications Technology Planning \& Evaluation).}

\blfootnote{\noindent F. Rostami Ghadi and K. K. Wong are with the Department of Electronic and Electrical Engineering, University College London, WC1E 7JE London, United Kingdom. K. K. Wong is also affiliated with the Department of Electronic Engineering, Kyung Hee University, Yongin-si, Gyeonggi-do 17104, Republic of Korea (e-mail: $\rm \{f.rostamighadi, kai\text{-}kit.wong\}@ucl.ac.uk$).}
\blfootnote{\noindent M. Kaveh is with the Department of Information and Communication Engineering, Aalto University, Espoo, Finland. (e-mail: $\rm masoud.kaveh@aalto.fi$).}
\blfootnote{\noindent M. J. Ahmadi is with the Chair of Information Theory and Machine Learning and with the Cluster of Excellence \textit{``Centre for Tactile Internet with Human-in-the-Loop (CeTI),''} Technische Universit\"at Dresden, 01062 Dresden, Germany (e-mail: $\rm mohammad\_javad.ahmadi@tu\text{-}dresden.de$).}
\blfootnote{K. F. Tong is with the School of Science and Technology, Hong Kong Metropolitan University, Hong Kong SAR, China (e-mail: $\rm  ktong@hkmu.edu.hk$).}
\blfootnote{H. Shin is with the Department of Electronics and Information Convergence Engineering, Kyung Hee University, Yongin-si, Gyeonggi-do 17104, Republic of Korea (e-mail: $\rm hshin@khu.ac.kr$).}

\blfootnote{\noindent \em Corresponding Author: Kai-Kit Wong.}
	
\begin{abstract}
Enormous fluid antenna systems (E-FAS) have recently emerged as a surface-wave (SW)-enabled architecture that can induce controllable large-scale channel gains through guided electromagnetic routing. This paper develops a secrecy analysis framework for E-FAS-assisted downlink transmission with practical pilot-based channel estimation. We consider a multiple-input single-output (MISO) wiretap setting in which the base station (BS) performs minimum mean-square-error (MMSE) channel estimation and adopts maximum-ratio transmission (MRT) with artificial noise (AN). To capture the leakage of SW routing in E-FAS, we introduce a correlated SW-leakage model that accounts for statistical coupling between the legitimate and eavesdropper channels caused by partially overlapping SW propagation paths. Exploiting the two-timescale nature---with slowly varying routing gain and small-scale block fading, we then derive a closed-form conditional expression for the secrecy outage probability (SOP) and a tractable characterization of the ergodic secrecy rate (ESR) in the presence of correlated quadratic forms. Our analysis yields three key insights: (i)~secrecy collapses at high transmit power if and only if AN is not present, whereas any strictly positive AN can prevent asymptotic collapse; (ii)~the optimal data-AN power split is achieved by a strictly interior solution; and (iii)~routing gain improves both the received signal strength and the channel-estimation quality, creating a nonlinear coupling that raises the signal-to-interference plus noise ratio (SINR) ceiling in the high signal-to-noise ratio (SNR) regime, and disperses secrecy across routing states. Numerical results indicate that E-FAS markedly enlarges the secure operating region significantly when compared with conventional space-wave transmission.
\end{abstract}


\vspace{-2mm}
\section{Introduction}
\IEEEPARstart{T}{he explosive} growth of wireless data traffic, together with increasingly strict requirements on system performance, is accelerating the evolution of mobile communication technologies toward sixth-generation (6G) networks. Emerging applications such as immersive extended-reality services, holographic communications, and large-scale Internet-of-Things (IoT) deployments require unprecedented levels of spectral efficiency, latency, reliability, and security \cite{Tariq-2020,mitev2023what}. In particular, ensuring confidential transmission over the wireless medium has become increasingly important due to the broadcast nature of radio propagation. While conventional cryptographic techniques provide security at higher layers, physical layer security offers an information-theoretic complement by exploiting the inherent randomness of wireless channels \cite{shiu2011phy}.

Though multi-antenna technologies have greatly improved wireless performance, conventional multiple-input multiple-output (MIMO) systems rely upon static antenna deployments and fixed propagation environments \cite{Gesbert-2007,wang2024tut}. This limitation motivates the exploration of new communication paradigms that introduce additional spatial degrees of freedom (DoF) in the physical layer. In this regard, the fluid antenna system (FAS) concept has recently emerged as a new communication architecture that empowers spatial reconfigurability by dynamically adjusting the antenna position within a confined region \cite{wong2021fluid,wong2020perfrom}. In essence, FAS represents a broad class of wireless communication systems that integrate reconfigurable antennas into the physical layer for system optimization \cite{new2025tut,FAS_survey_Hong,Lu-2025,FAS_enabler,FAS_wu_tuo1}. In recent years, FAS has also attracted attention in physical layer security, as spatial reconfigurability can enhance both communication reliability and secrecy \cite{ghadi2024phys,tang2023fluid,kaveh2026phy}.

Building upon this concept, a surface wave (SW)-enabled architecture that utilizes electromagnetic signal routing across programmable metasurfaces for energy-efficient interference-less communication, has been proposed \cite{wong2021vision}. This paradigm is referred to as enormous FAS (E-FAS) \cite{wong2025efas}. Specifically, in this paradigm, incident space waves are converted into guided surface waves and routed through engineered pathways before being re-radiated toward the intended receiver. This mechanism introduces a distinctive two-timescale channel structure consisting of a slowly varying routing gain determined by the surface topology and a fast small-scale fading component that evolves across coherence blocks. The routing capability of E-FAS has been shown to provide substantial performance gains in terms of reliability and achievable rate \cite{ghadi2026efas,ghadi2026estim}. 

However, secure communication over E-FAS-assisted links remains largely unexplored. From the physical layer security perspective, the interplay between routing gain, channel estimation, and artificial noise (AN)  transmission introduces new challenges that are absent in conventional space-wave systems. Under pilot-based channel acquisition, imperfect channel state information (CSI) will lead to estimation errors that cause self-interference during beamforming \cite{nosrat2011mimo}. Additionally, guided SW propagation may induce statistical coupling between the legitimate and eavesdropper channels when their signals share partially common routing paths over the metasurface. Such SW-guided propagation mechanisms have been demonstrated in metasurface-based communication architectures, where incident space waves are converted into surface waves that propagate along the structure before re-radiating toward different directions \cite{arshed2025surface}. These characteristics fundamentally alter the secrecy behavior of E-FAS-assisted systems \cite{ghadi2021copula}. 

\vspace{-2mm} 
\subsection{Motivation and Contributions}
In contrast to conventional models where large-scale fading mainly scales the signal-to-interference plus noise ratio (SINR), the routing gain in E-FAS simultaneously enhances the received signal strength and improves channel estimation quality. This dual effect modifies the effective SINR ceiling and changes the trade-off between information transmission and AN protection. As a consequence, a dedicated secrecy analysis that explicitly captures routing-induced randomness, channel estimation imperfections, and correlated leakage is required. Motivated by these observations, this paper develops a comprehensive secrecy framework for E-FAS-assisted downlink transmission under practical pilot-based channel estimation. Specifically, we consider a multiple-input single-output (MISO) system in which the base station (BS) performs minimum mean square error (MMSE) channel estimation and employs maximum ratio transmission (MRT) combined with AN. The eavesdropper's instantaneous CSI is assumed unknown at the BS transmitter. Nonetheless, due to shared SW propagation paths, the legitimate  and eavesdropper channels may exhibit statistical correlation, which is captured through the adopted channel model. By exploiting the two-timescale structure of the E-FAS channel, we then derive a closed-form conditional characterization of the secrecy outage probability (SOP) and obtain a tractable representation of the ergodic secrecy rate (ESR) under correlated quadratic forms. The main contributions of this paper are summarized as follows.
\begin{itemize}
\item We develop an analytical secrecy framework for E-FAS-assisted downlink transmission under statistically correlated leakage between the legitimate and eavesdropper channels. In contrast to conventional analyses based on independent fading assumptions, the proposed model captures correlation induced by guided SW propagation. The resulting SINR expressions lead to ratios of quadratic forms in complex Gaussian variables arising from the joint effects of SW propagation and AN injection.
\item We derive closed-form conditional expressions for the probability density function (PDF) and cumulative distribution function (CDF) of the eavesdropper SINR conditioned on the estimated legitimate channel. The results capture the eigen-structure of the AN subspace and provide a statistical characterization of the leakage channel.
\item Building on the derived conditional SINR distribution, we obtain an exact conditional representation of the secrecy SOP under imperfect CSI and AN transmission. We also derive a tractable representation of the ESR that captures the two-timescale nature of E-FAS channels, where routing gain evolves slowly while small-scale fading varies within each channel coherence block. 
\item We establish a fundamental structural property of the E-FAS secrecy performance at high transmit power. In particular, secrecy collapse occurs if and only if  AN is absent, whereas any strictly positive AN will eliminate asymptotic collapse. Additionally, the optimal data-AN power allocation admits a strictly interior solution.
\item We illustrate that the E-FAS routing gain influences secrecy performance through a nonlinear dual mechanism. Specifically, the routing gain simultaneously enhances the received signal power and improves channel estimation quality, thereby increasing the effective SINR ceiling in the high signal-to-noise ratio (SNR) regime, and enlarging the secure operating region.
\item Monte Carlo simulations validate the accuracy of the analytical results and provide detailed performance insights. The results demonstrate that E-FAS-assisted transmission significantly enlarges the secure operating region compared with conventional space-wave systems, while revealing the impact of routing gain, AN power allocation, and channel correlation on secrecy performance.
\end{itemize}


\vspace{-2mm}
\section{Secure System Model}\label{sec:system_model_sec}
As shown in Fig.~\ref{fig:system}, we consider a downlink secure transmission scenario in which an $M$-antenna BS (Alice) transmits confidential information to a single-antenna legitimate user (Bob) in the presence of a passive single-antenna eavesdropper (Eve). The transmission is assisted by an E-FAS deployed in the propagation environment.\footnote{In this paper, all terminals are equipped with fixed-position antennas except the metasurface is treated as an E-FAS.} We assume that the direct space-wave Alice-Bob link is severely attenuated due to environmental blockages, e.g., walls or surrounding obstacles, such that signal propagation is dominated by the E-FAS-assisted SW route. In particular, the E-FAS consists of programmable metasurface tiles deployed on surrounding structures such as walls, ceilings, or facades, which convert incident space waves into guided SWs, route electromagnetic energy along engineered surface pathways, and re-radiate the signal toward Bob through programmable launcher structures \cite{wong2025efas}. We further assume time-division duplexing (TDD) operation and uplink pilot-based channel estimation. Alice acquires an MMSE estimate of Bob's channel and designs its downlink beamforming based solely on this estimate, while the instantaneous CSI of Eve is unavailable at Alice. Additionally, to enhance secrecy, Alice employs beamforming combined with AN transmission. We assume that the E-FAS-induced channel exhibits a two-timescale structure consisting of a slowly varying large-scale routing gain and a fast small-scale fading component within each channel coherence block. Furthermore, because of shared SW propagation pathways, the effective channels of Bob and Eve would be statistically coupled, as will be modeled through the correlated small-scale channel components to be introduced next. The large-scale routing gain of Eve is treated as an unknown but bounded parameter, reflecting uncertainty in the eavesdropper's location and routing configuration.

\begin{figure}[]
\centering
\includegraphics[width=.8\columnwidth]{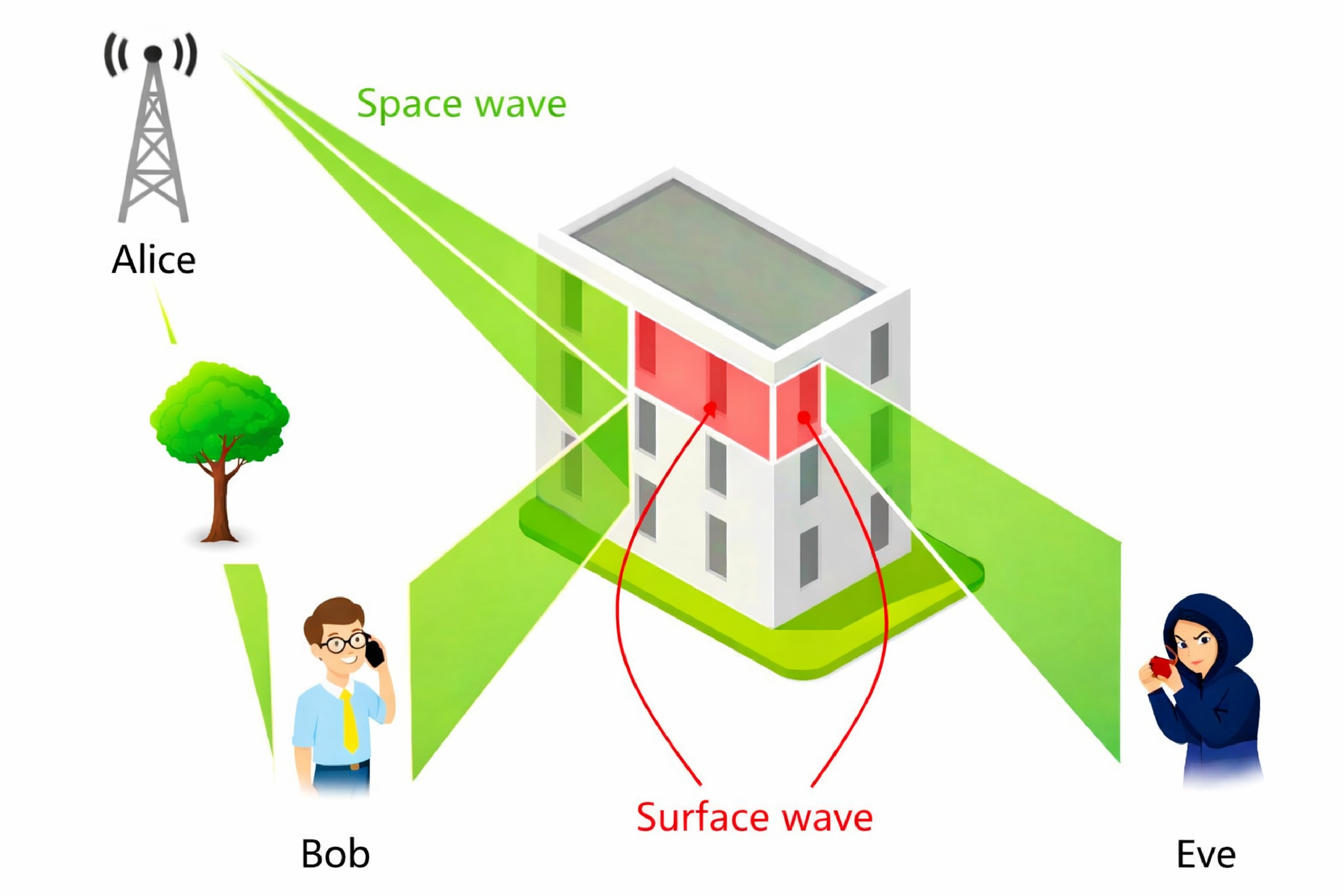}
\caption{The E-FAS-assisted secure communication scenario.}\label{fig:system}
\vspace{-2mm}
\end{figure}

\vspace{-2mm}
\subsection{Two-Timescale Channel and Routing-Gain Randomness}
The E-FAS-assisted BS-Bob channel is modeled using the two-timescale equivalent representation as \cite{emil2017massive}
\begin{align}\label{eq:hb_sec_B}
\mathbf{h}_b = \sqrt{\beta_b}\,\mathbf{R}_b^{1/2}\mathbf{g}_b,
\end{align}
where $\mathbf{g}_b\sim \mathcal{CN}(\mathbf{0},\mathbf{I}_M)$ denotes the small-scale fading vector, $\mathbf{R}_b\succeq \mathbf{0}$ is a deterministic spatial correlation matrix capturing the effective array-side correlation induced by the composite E-FAS propagation, and $\beta_b>0$ is the routing-dependent large-scale gain. The routing gain $\beta_b$ aggregates the effects of SW attenuation along the configured pathway(s), junction and launcher losses, and routing geometry. In contrast to conventional space-wave models where large-scale fading is purely geometry-induced and typically treated as constant over many coherence blocks, $\beta_b$ in E-FAS depends on the programmable routing configuration and thus evolves on a slow routing timescale $T_r\gg T_c$ \cite{ghadi2026efas}, so that the routing gain can be treated as constant over multiple coherence blocks.

To explicitly activate the two-timescale structure in the secrecy analysis, we model $\beta_b$ as a random variable over the routing timescale with distribution $f_{\beta_b}(\cdot)$ supported on $(0,\infty)$. Conditioning on $\beta_b$ yields a block-fading model over coherence blocks, whereas averaging over $f_{\beta_b}(\cdot)$ captures the secrecy variability across routing states. The distribution $f_{\beta_b}(\cdot)$ may be obtained from geometry-based routing statistics or treated parametrically; the ensuing analysis applies to any distribution with finite first and second moments.

\vspace{-2mm}
\subsection{Eavesdropper Channel with SW Leakage Coupling}
Similarly, we model Eve's channel as
\begin{align}\label{eq:he_sec_B}
\mathbf{h}_e = \sqrt{\beta_e}\,\mathbf{R}_e^{1/2}\mathbf{g}_e,
\end{align}
where $\mathbf{R}_e\succeq \mathbf{0}$ is Eve's statistical spatial correlation matrix, $\beta_e>0$ is the large-scale routing gain of the BS-Eve link due to the E-FAS, and $\mathbf{g}_e\sim\mathcal{CN}(\mathbf{0},\mathbf{I}_M)$. 

Unlike the independent fading assumption in conventional secrecy baselines, the E-FAS setting may induce statistical coupling between Bob and Eve when they share a portion of the same routed SW structure. To obtain a tractable representation, the leakage-induced dependence is modeled at the small-scale fading level, while the large-scale gains $\beta_b$ and $\beta_e$ are kept as separate routing-dependent parameters. Accordingly, we adopt the correlated jointly Gaussian model \cite{zhou2010secure}
\begin{align}\label{eq:ge_corr_model}
\mathbf{g}_e = \rho\,\mathbf{C}\mathbf{g}_b + \sqrt{1-\rho^2}\,\mathbf{u},
\end{align}
where $\rho\in[0,1)$ quantifies the SW leakage coupling strength, $\mathbf{C}\in\mathbb{C}^{M\times M}$ denotes a deterministic coupling matrix satisfying $\|\mathbf{C}\|\le 1$, and $\mathbf{u}\sim\mathcal{CN}(\mathbf{0},\mathbf{I}_M)$ is independent of $\mathbf{g}_b$. This model implies
$
\mathbb{E}\{\mathbf{g}_e\mathbf{g}_b^H\}=\rho\mathbf{C}
$, 
and reduces to independent fading when $\rho=0$. The parameter $\rho$ quantifies the statistical coupling between the channels of Bob and Eve, which may arise when their signals traverse partially shared SW routing structures. A larger $\rho$ implies stronger statistical coupling between their channels and potentially increases information leakage. Note that this model provides an effective statistical representation of the channel dependence induced by partially shared SW routing structures rather than a first-principles derivation of the full E-FAS propagation network.

\vspace{-2mm}
\subsection{Uncertainty Model for Eve's Large-Scale Gain}
We assume that Alice has no knowledge of Eve's instantaneous CSI. Due to the passive nature of the eavesdropper, the large-scale gain $\beta_e$ is also not perfectly known at the BS. We thus model $\beta_e$ as belonging to a bounded uncertainty set, allowing a worst-case secrecy characterization \cite{mukher2011rob}, i.e.,
\begin{align}\label{eq:betae_uncertainty}
\beta_e \in \mathcal{B}_e \triangleq [\beta_e^{\min},\,\beta_e^{\max}],
\end{align}
where $\beta_e^{\min}$ and $\beta_e^{\max}$ denote deterministic bounds determined by the propagation environment and deployment geometry. Secrecy metrics are characterized either conditionally on a given $\beta_e$ or in a worst-case sense by evaluating the supremum over the uncertainty set $\mathcal{B}_e$. This formulation avoids relying on perfect knowledge of Eve's large-scale fading and enables a robust secrecy characterization under large-scale uncertainty. Also, it is worth noting that the large-scale gains $\beta_b$ and $\beta_e$ may differ due to distinct radiation and propagation conditions from the SW structure toward Bob and Eve.

\vspace{-2mm}
\subsection{Uplink Training and Estimation of Bob's Channel}
Channel estimation is performed via uplink pilot transmission exploiting TDD reciprocity. Within each coherence block of length $T_c$, Bob transmits a pilot sequence of length $\tau_p$ symbols with per-symbol power $\rho_p$. As a result, the BS obtains the following sufficient statistic
\begin{align}\label{eq:pilot_obs_B}
\mathbf{y}_p = \sqrt{\tau_p\rho_p}\,\mathbf{h}_b + \mathbf{n}_p,
\end{align}
where $\mathbf{n}_p \sim \mathcal{CN}(\mathbf{0},\sigma^2\mathbf{I}_M) \in \mathbb{C}^{M\times 1}$ is the post-processing additive white Gaussian noise vector at the BS during the uplink training phase, independent of $\mathbf{h}_b$. Now, conditioned on $(\beta_b,\mathbf{R}_b)$, the linear MMSE estimate of $\mathbf{h}_b$ is \cite{marz2016fun}
\begin{align}\label{eq:mmse_hb_corr}
\hat{\mathbf{h}}_b=\sqrt{\tau_p\rho_p}\,\mathbf{R}_{h_b}
\left(\tau_p\rho_p\,\mathbf{R}_{h_b}+\sigma^2\mathbf{I}_M\right)^{-1}\mathbf{y}_p,
\end{align}
where $\mathbf{R}_{h_b} \triangleq \mathbb{E}\{\mathbf{h}_b \mathbf{h}_b^H\} = \beta_b \mathbf{R}_b$, and the estimation error is defined as $\tilde{\mathbf{h}}_b = \mathbf{h}_b - \hat{\mathbf{h}}_b$. Owing to the linear MMSE estimation under jointly Gaussian statistics, the estimate $\hat{\mathbf{h}}_b$ and the error $\tilde{\mathbf{h}}_b$ are statistically independent and satisfy
\begin{align}
\mathbf{R}_{\hat{h}}&\triangleq\mathbb{E}\{\hat{\mathbf{h}}_b\hat{\mathbf{h}}_b^H\}
=\tau_p\rho_p\,\mathbf{R}_{h_b}\left(\tau_p\rho_p\,\mathbf{R}_{h_b}+\sigma^2\mathbf{I}_M\right)^{-1}\mathbf{R}_{h_b},\label{eq:Rhat_general}\\
\mathbf{R}_{\tilde{h}}&\triangleq\mathbb{E}\{\tilde{\mathbf{h}}_b\tilde{\mathbf{h}}_b^H\}
=\mathbf{R}_{h_b}-\mathbf{R}_{\hat{h}}.\label{eq:Rtil_general}
\end{align}
The above covariance expressions degenerate to the spatially uncorrelated case when $\mathbf{R}_b=\mathbf{I}_M$. In this case, both the channel estimate and the estimation error have isotropic covariance structures, i.e.,
$
\mathbf{R}_{\hat h_b} = \Omega_{\hat h} \mathbf{I}_M$ and 
$\mathbf{R}_{\tilde h_b} = \Omega_{\tilde h} \mathbf{I}_M
$, 
in which the scalar variances are, respectively, given by
\begin{align}
\Omega_{\hat h}&=\frac{\tau_p \rho_p \beta_b^2}{\tau_p \rho_p \beta_b + \sigma^2}, \label{eq:om1}\\
\Omega_{\tilde h}&=\beta_b - \Omega_{\hat h}= \frac{\beta_b\sigma^2}{\tau_p\rho_p\beta_b+\sigma^2}.\label{eq:om2}
\end{align}

Importantly, for any finite pilot energy $\tau_p \rho_p < \infty$ and noise variance $\sigma^2 > 0$, the estimation error covariance $\mathbf{R}_{\tilde h_b}$ is strictly positive definite. Hence, perfect CSI is unattainable under practical training conditions, and residual estimation-induced distortion is inherent in the downlink transmission.

\vspace{-2mm}
\subsection{Downlink Signaling with AN}
To mitigate the high-SNR secrecy degradation associated with MRT under imperfect CSI, Alice employs AN transmission \cite{goel2008guar}. The transmit signal is constructed as a superposition of an information-bearing beam and an AN component, i.e.,
\begin{equation}\label{eq:tx_an_model}
\mathbf{x}=\sqrt{P_s}\mathbf{w}s+\sqrt{P_a}\mathbf{V}\mathbf{z},
\end{equation}
where $s$ represents the unit-power information symbol, $\mathbf{z} \sim \mathcal{CN}(\mathbf{0},\mathbf{I}_{M-1})$ is the AN vector independent of $s$, $P_s$ and $P_a$ denote the power allocated to the information signal and AN, respectively, satisfying $P_s + P_a = P$, and  $\mathbf{V} \in \mathbb{C}^{M \times (M-1)}$ is the AN shaping matrix.

The information beamformer is chosen according to MRT based on the channel estimate, i.e.,
\begin{equation}\label{eq:w_mrt_B}
\mathbf{w} = \frac{\hat{\mathbf{h}}_b}{\|\hat{\mathbf{h}}_b\|}.
\end{equation}
Since the beamformer $\mathbf w$ is a deterministic function of the channel estimate $\hat{\mathbf h}_b$, conditioning on $\hat{\mathbf h}_b$ renders $\mathbf w$ deterministic in the subsequent statistical analysis. Also, the AN shaping matrix $\mathbf{V} \in \mathbb{C}^{M \times (M-1)}$ forms an orthonormal basis for the null space of $\hat{\mathbf{h}}_b$, satisfying $\mathbf{V}^H \mathbf{V} = \mathbf{I}_{M-1}$ and $\mathbf{V}^H \hat{\mathbf{h}}_b = \mathbf{0}$.

Under perfect CSI, this construction eliminates AN interference at Bob. Under imperfect CSI, however, the mismatch between $\hat{\mathbf{h}}_b$ and the true channel $\mathbf{h}_b$ results in a residual AN leakage term at Bob, whose power is governed by the estimation error statistics. Therefore, the received signals at Bob and Eve are, respectively, defined as
\begin{align}
y_b &= \mathbf{h}_b^H\mathbf{x} + n_b,\label{eq:yb_B}\\
y_e &= \mathbf{h}_e^H\mathbf{x} + n_e,\label{eq:ye_B}
\end{align}
where $n_b,n_e\sim\mathcal{CN}(0,\sigma^2)$ are independent noise terms. 

\vspace{-2mm}
\subsection{Secrecy Metrics and Robust SOP}
Let $\gamma_b$ and $\gamma_e$ denote the instantaneous effective SINRs at Bob and Eve, respectively, induced by the transmitted signal in \eqref{eq:tx_an_model}. Accounting for the pilot overhead, the effective data transmission fraction within each coherence block is given by
$
\eta \triangleq 1 - \frac{\tau_p}{T_c}
$.
Therefore, the instantaneous achievable rates at Bob and Eve are, respectively, defined as
\begin{align}
R_b &= \eta \log_2(1+\gamma_b),\\
R_e &= \eta \log_2(1+\gamma_e),
\end{align}
where $\gamma_b$ and $\gamma_e$ are the SINRs at Bob and Eve, respectively, and their expressions will be given later. Hence, the instantaneous secrecy rate is given by
\begin{align}\label{eq:rs}
R_s = \big[ R_b - R_e \big]^+,
\end{align}
where $[x]^+=\max\{0,x\}$. For a target secrecy rate $R_{\mathrm{th}}$, the SOP is defined as
\begin{equation}\label{eq:sop}
P_{\mathrm{so}}(\beta_e,\alpha)=\Pr\!\left( R_s < R_{\mathrm{th}} \right),
\end{equation}
where the dependence on the eavesdropper large-scale routing gain $\beta_e$ and the power allocation ratio $\alpha \triangleq P_s/P$ is made explicit. Under the large-scale uncertainty model in \eqref{eq:betae_uncertainty}, the worst-case SOP can be defined as \cite{mukher2011rob}
\begin{align}\label{eq:sop_worst_case}
P_{\mathrm{so}}^{\mathrm{wc}}(\alpha)\triangleq\sup_{\beta_e \in \mathcal{B}_e}P_{\mathrm{so}}(\beta_e,\alpha),
\end{align}
which can quantify the secrecy performance in a robustness-oriented manner without assuming perfect knowledge of Eve's large-scale channel strength. 

\vspace{-2mm}
\section{Statistical Characterization of  SINRs}\label{sec:sinr_analysis_sec}
In this section, we derive the instantaneous SINRs at Bob and Eve under correlated SW leakage and AN transmission, where due to the shared propagation environment and the power-splitting parameter $\alpha$, the resulting SINRs are statistically coupled through the routing-dependent channel statistics.

\vspace{-2mm}
\subsection{Legitimate SINR Under Imperfect CSI and AN}
By utilizing \eqref{eq:tx_an_model}, the received signal at Bob is rewritten as
\begin{align}
y_b=\sqrt{P_s}\mathbf{h}_b^H\mathbf{w}s+\sqrt{P_a}\mathbf{h}_b^H\mathbf{V}\mathbf{z}+n_b.
\end{align}
Substituting the channel decomposition $\mathbf{h}_b=\hat{\mathbf{h}}_b+\tilde{\mathbf{h}}_b$ yields
\begin{align}
y_b=\sqrt{P_s}\hat{\mathbf{h}}_b^H\mathbf{w}s+\sqrt{P_s}\tilde{\mathbf{h}}_b^H\mathbf{w}s+\sqrt{P_a}\tilde{\mathbf{h}}_b^H\mathbf{V}\mathbf{z}+n_b,
\end{align}
where the term $\mathbf{V}^H\hat{\mathbf{h}}_b=0$ eliminates any AN leakage under perfect CSI. Then conditioned on the channel estimate $\hat{\mathbf{h}}_b$, the first term represents the desired signal. Furthermore, since (\ref{eq:w_mrt_B}), its instantaneous power equals
\begin{align}
P_s|\hat{\mathbf{h}}_b^H\mathbf{w}|^2=P_s\|\hat{\mathbf{h}}_b\|^2,
\end{align}
and the remaining terms constitute effective interference due to channel estimation errors. Specifically,
\begin{align}
I_{b,1}&=P_s|\tilde{\mathbf{h}}_b^H\mathbf{w}|^2,\\
I_{b,2}&=P_a\|\tilde{\mathbf{h}}_b^H\mathbf{V}\|^2,
\end{align}
correspond to self-interference from imperfect beamforming and AN leakage induced by CSI mismatch, respectively. Also, since $\tilde{\mathbf{h}}_b \sim \mathcal{CN}(\mathbf{0},\mathbf{R}_{\tilde h_b})$ is independent of $\hat{\mathbf{h}}_b$, the conditional second-order statistics are given by
\begin{align}
\mathbb{E}[I_{b,1}|\hat{\mathbf{h}}_b]&=P_s\mathbf{w}^H\mathbf{R}_{\tilde h_b}\mathbf{w},\\
\mathbb{E}[I_{b,2}|\hat{\mathbf{h}}_b]&=P_a{\rm tr}(\mathbf{V}^H\mathbf{R}_{\tilde h_b}\mathbf{V}).
\end{align}
Adopting the standard use-and-forget bound \cite{ngo2013energy}, these conditional interference powers are treated as deterministic effective noise. Hence, the effective SINR at Bob can be found as
\begin{align}\label{eq:gamma_b_general}
\gamma_b=\frac{P_s\|\hat{\mathbf{h}}_b\|^2}{P_s\mathbf{w}^H\mathbf{R}_{\tilde h_b}\mathbf{w}+P_a\operatorname{tr}\!\left(\mathbf{V}^H\mathbf{R}_{\tilde h_b}\mathbf{V}\right)+\sigma^2}.
\end{align}

In the spatially uncorrelated case where 
$\mathbf{R}_{\tilde h_b}=\Omega_{\tilde h}\mathbf{I}_M$, 
\eqref{eq:gamma_b_general} simplifies to
\begin{align}
\gamma_b=\frac{P_s\|\hat{\mathbf{h}}_b\|^2}{\Omega_{\tilde h}\!\left(P_s + P_a(M-1)\right)+\sigma^2}.
\end{align}
This expression reveals two distinct degradation mechanisms at Bob.  The numerator term, proportional to $P_s$, corresponds to residual self-interference due to imperfect beamforming.  The denominator term, proportional to $P_a(M-1)$, quantifies AN leakage caused by CSI estimation errors. Therefore, while AN degrades Eve's reception, it also introduces a leakage-induced penalty at Bob under the imperfect CSI case.

\vspace{-2mm}
\subsection{Eavesdropper SINR Under SW Leakage Correlation}
By \eqref{eq:tx_an_model}, the received signal at Eve is reformulated as
\begin{align}
y_e=\sqrt{P_s}\mathbf{h}_e^H\mathbf{w}s+\sqrt{P_a}\mathbf{h}_e^H\mathbf{V}\mathbf{z}+n_e.
\end{align}
Thus, the instantaneous SINR at Eve is given by
\begin{equation}\label{eq:gamma_e_general}
\gamma_e=\frac{P_s|\mathbf{h}_e^H\mathbf{w}|^2}{P_a\|\mathbf{h}_e^H\mathbf{V}\|^2+\sigma^2}.
\end{equation}
Under the correlated SW leakage model in \eqref{eq:ge_corr_model}, we have
\begin{equation}
\mathbf{h}_e=\sqrt{\beta_e}\mathbf{R}_e^{1/2}\left(\rho\mathbf{C}\mathbf{g}_b+\sqrt{1-\rho^2}\mathbf{u}\right),
\end{equation}
where $\mathbf{u} \sim \mathcal{CN}(\mathbf{0},\mathbf{I}_M)$ is independent of $\mathbf{g}_b$. 

Given that the beamformer $\mathbf{w}$ is a deterministic function of $\hat{\mathbf{h}}_b$, and $\hat{\mathbf{h}}_b$ depends linearly on $\mathbf{g}_b$, the term $\mathbf{h}_e^H\mathbf{w}$ is statistically coupled with $\hat{\mathbf{h}}_b$ through the correlation parameter $\rho$. Also, conditioned on $\hat{\mathbf{h}}_b$, the random vector $\mathbf{u}$ remains independent, and $\mathbf{g}_b$ has a Gaussian conditional distribution. Therefore, $\mathbf{h}_e^H\mathbf{w}$ remains conditionally complex Gaussian and its conditional mean and variance are, respectively, given by
\begin{align}
\mathbb{E}\!\left[\mathbf{h}_e^H\mathbf{w}\mid\hat{\mathbf{h}}_b\right]&=\rho\sqrt{\beta_e}\mathbf{w}^H\mathbf{R}_e^{1/2}\mathbf{C}\mathbb{E}\!\left[\mathbf{g}_b\mid\hat{\mathbf{h}}_b\right],\\
\operatorname{Var}\!\left[\mathbf{h}_e^H\mathbf{w}\mid\hat{\mathbf{h}}_b\right]&=\beta_e(1-\rho^2)\mathbf{w}^H\mathbf{R}_e\mathbf{w}.
\end{align}
Accordingly, conditioned on $\hat{\mathbf{h}}_b$, the numerator term $|\mathbf{h}_e^H\mathbf{w}|^2$ follows a noncentral chi-square distribution with one complex DoF \cite{simon2004digit} and the noncentrality parameter depends explicitly on the routing-induced correlation coefficient $\rho$.

In this regard, the denominator term $\|\mathbf{h}_e^H\mathbf{V}\|^2$ can be expressed as a quadratic form in correlated Gaussian variables, namely
$
\|\mathbf{h}_e^H\mathbf{V}\|^2
=
\mathbf{h}_e^H\mathbf{V}\mathbf{V}^H\mathbf{h}_e
$,
whose distribution is determined by the eigenvalues of 
$\mathbf{R}_e^{1/2}\mathbf{V}\mathbf{V}^H\mathbf{R}_e^{1/2}$. Consequently, $\gamma_e$ is the ratio of correlated quadratic forms in complex Gaussian variables. In contrast to the independent Rayleigh case, it does not admit an exponential distribution and its conditional distribution can be characterized via eigenvalue decomposition of the involved covariance matrices.

\begin{remark}[Statistical Coupling Between $\gamma_b$ and $\gamma_e$]
Under the correlated SW leakage model, the eavesdropper channel contains a component aligned with the legitimate channel through the term $\rho \mathbf{C}\mathbf{g}_b$. As the beamformer $\mathbf{w}$ is constructed from the channel estimate $\hat{\mathbf{h}}_b$, which itself depends on $\mathbf{g}_b$, both $\gamma_b$ and $\gamma_e$ inherit statistical dependence on the common underlying random vector $\mathbf{g}_b$ whenever $\rho > 0$. As a result, $\gamma_b$ and $\gamma_e$ are not statistically independent in general. The independence structure is recovered as a special case when $\rho = 0$, for which the eavesdropper channel becomes independent of $\mathbf{g}_b$ and thus independent of the beamforming direction. In this case, $\gamma_b$ and $\gamma_e$ are statistically independent.
\end{remark}

\vspace{-2mm}
\section{Robust SOP Under Correlated Leakage}\label{sec:sop_analysis}
This section characterizes the SOP under correlated SW leakage, AN transmission, routing randomness, and large-scale uncertainty in the eavesdropper link. 

\vspace{-2mm}
\subsection{Outage Event Reformulation}
From \eqref{eq:rs}, the instantaneous secrecy rate is defined as
\begin{align}
R_s=\eta\left[\log_2(1+\gamma_b)-\log_2(1+\gamma_e)\right]^+.
\end{align}
For a target secrecy spectral efficiency $R_{\mathrm{th}}$, we define $\theta \triangleq 2^{R_{\mathrm{th}}/\eta} > 1$. Using the SOP definition in \eqref{eq:sop}, we have
\begin{align}
P_{\mathrm{so}}(\beta_e,\rho)&=\Pr\!\left(\gamma_b < \theta(1+\gamma_e) - 1\right)\notag\\
&=\iint_{\mathcal{D}}f_{\gamma_b,\gamma_e}(x,y)\,dx\,dy,
\end{align}
where $\mathcal{D} = \{(x,y): x < \theta(1+y) - 1\}$. Under correlated SW leakage, $\gamma_b$ and $\gamma_e$ are statistically dependent through the common routing-induced randomness. Hence, the SOP is expressed in terms of their joint PDF $f_{\gamma_b,\gamma_e}(x,y)$. 

\vspace{-2mm}
\subsection{Conditional SOP Representation}
To handle the statistical coupling between $\gamma_b$ and $\gamma_e$, we condition on the channel estimate $\hat{\mathbf{h}}_b$. Using the law of total probability, the SOP can be expressed as
\begin{align}\label{eq:sop_conditional_master}
P_{\mathrm{so}}=\mathbb{E}_{\hat{\mathbf{h}}_b,\beta_b,\beta_e}\left[\Pr\left(\gamma_b < \theta(1+\gamma_e)-1\mid\hat{\mathbf{h}}_b,\beta_b,\beta_e\right)\right].
\end{align}
Conditioned on $\hat{\mathbf{h}}_b$, the legitimate SINR in \eqref{eq:gamma_b_general} reduces to
\begin{equation}
\gamma_b=\frac{P_s\|\hat{\mathbf{h}}_b\|^2}{\Xi_b(\hat{\mathbf{h}}_b)},
\end{equation}
where
\begin{align}
\Xi_b(\hat{\mathbf{h}}_b)=P_s\mathbf{w}^H\mathbf{R}_{\tilde{h}}\mathbf{w}+P_a\operatorname{tr}(\mathbf{V}^H\mathbf{R}_{\tilde{h}}\mathbf{V})+\sigma^2.
\end{align}
Given that $\hat{\mathbf{h}}_b$ fully determines $\mathbf{w}$ and $\mathbf{V}$, the quantity $\gamma_b$ becomes deterministic under this conditioning. In contrast, the eavesdropper SINR in \eqref{eq:gamma_e_general} remains random and is given by
\begin{align}\label{eq:ge_XY_def}
\gamma_e=\frac{P_s |X|^2}{P_a Y + \sigma^2},
\end{align}
where
$X= \mathbf{h}_e^H \mathbf{w}$ and $Y = \mathbf{h}_e^H \mathbf{V}\mathbf{V}^H \mathbf{h}_e$.

Hence, under the correlated leakage model \eqref{eq:ge_corr_model} and conditioned on $\hat{\mathbf{h}}_b$, the random variable $X$ is complex Gaussian with the following conditional mean and variance
\begin{align}
\mu_X(\hat{\mathbf{h}}_b)&=\rho \sqrt{\beta_e}\mathbf{w}^H \mathbf{R}_e^{1/2}\mathbf{C}\mathbb{E}[\mathbf{g}_b \mid \hat{\mathbf{h}}_b],\\
\sigma_X^2(\hat{\mathbf{h}}_b)&=\beta_e(1-\rho^2)\mathbf{w}^H \mathbf{R}_e \mathbf{w},
\end{align}
where using the MMSE orthogonality principle, the conditional mean of the small-scale fading vector satisfies
$ \mathbb{E}[\mathbf g_b \mid \hat{\mathbf h}_b]
= \frac{1}{\sqrt{\beta_b}}\hat{\mathbf h}_b
$. 
Accordingly, $|X|^2$ follows a noncentral chi-square distribution with two real DoF (equivalently, one complex DoF) \cite{simon2004digit}. Also, the denominator term $Y$ is a quadratic form in correlated complex Gaussian variables, whose distribution is governed by the eigenvalues of 
$\mathbf{R}_e^{1/2}\mathbf{V}\mathbf{V}^H\mathbf{R}_e^{1/2}$ and therefore follows a generalized chi-square distribution.

Importantly, $X$ and $Y$ are generally statistically dependent due to the shared channel vector $\mathbf{h}_e$. Consequently, $\gamma_e$ is the ratio of correlated quadratic forms.

\vspace{-2mm}
\subsection{Closed-Form Conditional CDF of $\gamma_e$}\label{subsec:cond_cdf_ge}
To characterize the SOP, we derive the conditional CDF of $\gamma_e$ given $\hat{\mathbf h}_b$, which uniquely determines the beamformer $\mathbf w$ and the AN subspace $\mathbf V$. Under the correlated SW leakage model and the jointly Gaussian assumption, the eavesdropper channel conditioned on $\hat{\mathbf h}_b$ remains complex Gaussian, i.e.,
\begin{align}\label{eq:he_cond_gauss_fixed}
\mathbf h_e \mid \hat{\mathbf h}_b\sim\mathcal{CN}\!\big(\boldsymbol{\mu}_{e|b},\mathbf R_{e|b}\big),
\end{align}
where $\boldsymbol{\mu}_{e|b}$ captures the leakage-induced statistical coupling through $\rho$, and $\mathbf R_{e|b}$ denotes the conditional covariance matrix.

Now, conditioned on $\hat{\mathbf h}_b$, the scalar projection $X$ is complex Gaussian with mean and variance, given as
\begin{align}
\mu_X&\triangleq\mathbb E[X \mid \hat{\mathbf h}_b]=\mathbf w^H \boldsymbol{\mu}_{e|b},\\
\sigma_X^2&\triangleq\mathrm{Var}[X \mid \hat{\mathbf h}_b]=\mathbf w^H \mathbf R_{e|b} \mathbf w.
\end{align}
As a consequence, 
$
	X \mid \hat{\mathbf h}_b
	\sim
	\mathcal{CN}(\mu_X,\sigma_X^2)
$, 
and $|X|^2 \mid \hat{\mathbf h}_b$ follows a noncentral chi-square distribution with two real DoF (equivalently, a noncentral exponential distribution) \cite{simon2004digit}.

We now define the conditional covariance restricted to the AN subspace as
\begin{align}\label{eq:Q_def}
\mathbf{Q} \triangleq \mathbf{V}^H \mathbf{R}_{e|b}\mathbf{V} \in \mathbb{C}^{(M-1)\times(M-1)}.
\end{align}
Besides, we denote  $\{\lambda_i\}_{i=1}^{M-1}$ as the non-zero eigenvalues of $\mathbf{Q}$, such that $\mathbf{Q}=\mathbf{U}\boldsymbol{\Lambda}\mathbf{U}^H$ with $\boldsymbol{\Lambda}=\mathrm{diag}(\lambda_1,\ldots,\lambda_{M-1})$. Then, we have 
\begin{align}\label{eq:Y_eig_decomp_fixed}
Y \mid \hat{\mathbf h}_b=\sum_{i=1}^{M-1}\lambda_i |u_i|^2,
\end{align}
where $u_i \sim \mathcal{CN}(0,1)$ are independent. Hence, $Y \mid \hat{\mathbf h}_b$ follows a generalized chi-square (hypoexponential) distribution \cite{mathai1992qua}.

To derive a closed-form conditional CDF, we define the noncentrality parameter $\kappa$ and the respective constants as
\begin{align}
\kappa &\triangleq \frac{|\mu_X|^2}{\sigma_X^2},\label{eq:kappa_def}\\
\delta(t)& \triangleq \frac{\sigma^2 t}{P_s\sigma_X^2},\\
\xi(t) &\triangleq \frac{P_a t}{P_s\sigma_X^2}.\label{eq:delta_xi_def}
\end{align}
Now, assuming conditional independence between $X$ and $\mathbf{V}^H\mathbf{h}_e$ given $\hat{\mathbf h}_b$, which holds exactly when the conditional covariance $\mathbf R_{e|b}$ is isotropic and serves as an accurate approximation otherwise, and assuming that the eigenvalues $\{\lambda_i\}$ are pairwise distinct, the conditional CDF of $\gamma_e$ admits the following closed-form series representation for $t\ge0$:
\begin{align}
	F_{\gamma_e \mid \hat{\mathbf h}_b}(t)
	=
	1
	-
	e^{-\kappa-\delta(t)}
	\sum_{m=0}^{\infty}
	\frac{\kappa^m}{m!}
	\sum_{i=1}^{M-1}
	\alpha_i
	\left(1+\xi(t)\lambda_i\right)^{-(m+1)}\hspace{-2mm},
	\label{eq:conditional_ge_series_fixed}
\end{align}
where the partial-fraction coefficients $\alpha_i$ are given by
\begin{align}
	\alpha_i
	\triangleq
	\prod_{\substack{j=1\\ j\neq i}}^{M-1}
	\frac{\lambda_i}{\lambda_i-\lambda_j}.
	\label{eq:alpha_i_fixed}
\end{align}

It is noteworthy that for $\kappa=0$ (conditionally central numerator), the infinite series collapses to a finite sum, yielding a closed-form expression in rational-exponential form. 

\begin{remark}[Derivation Outline]
The expression \eqref{eq:conditional_ge_series_fixed} follows from conditioning on $Y$ and applying the Poisson-mixture representation of the noncentral chi-square distribution of $|X|^2$, together with the Laplace transform of $Y$, which is a generalized chi-square random variable. The result remains exact under eigenvalue multiplicities by replacing simple poles with higher-order terms. Besides, Fig.~\ref{fig:eve_cdf} illustrates the conditional CDF of the eavesdropper SINR $F_{\gamma_e|\hat{\mathbf{h}}_b}(t)$, given in \eqref{eq:conditional_ge_series_fixed}, for different values of the leakage correlation coefficient $\rho$, where the analytical curves closely match the Monte-Carlo simulation results, confirming the accuracy of the derived CDF expression. The detailed derivation follows standard steps and is omitted for brevity due to space limitations.
\end{remark}

\begin{figure}[]
\centering
\includegraphics[width=.95\columnwidth]{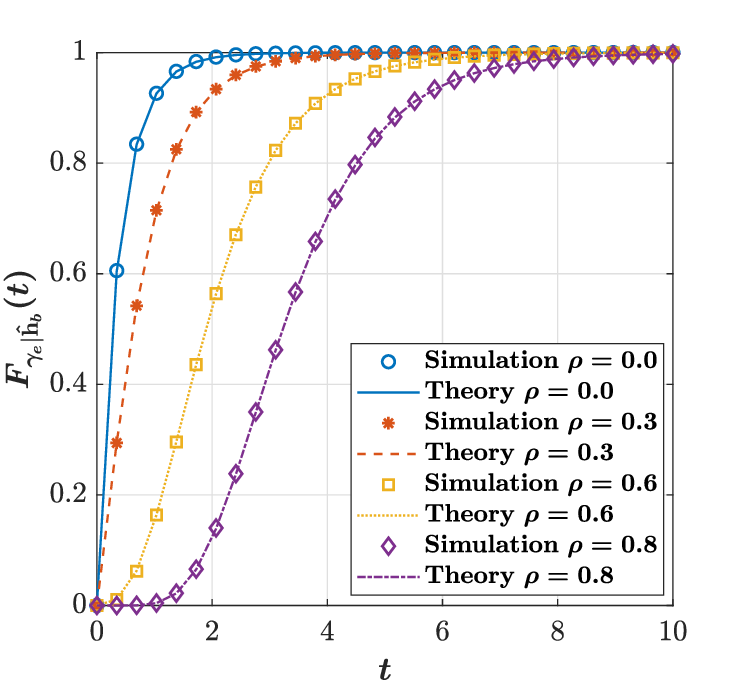}
\caption{Conditional CDF of the eavesdropper SINR under different channel correlation coefficients $\rho$.}\label{fig:eve_cdf}
\vspace{-2mm}
\end{figure}

\vspace{-2mm}
\subsection{Final SOP Expression}\label{subsec:sop_final}
As previously explained, by conditioning on $\hat{\mathbf{h}}_b$, the legitimate SINR $\gamma_b=\gamma_b(\hat{\mathbf{h}}_b)$ is deterministic, whereas $\gamma_e$ remains random with conditional CDF $F_{\gamma_e|\hat{\mathbf{h}}_b}(\cdot)$ derived previously. In this regard, we define the induced threshold as
\begin{align}\label{eq:tb_compact}
t_b(\hat{\mathbf{h}}_b)\triangleq \frac{\gamma_b(\hat{\mathbf{h}}_b)+1}{\theta}-1.
\end{align}
Hence, the conditional SOP reduces to
\begin{align}
P_{\mathrm{so}\mid\hat{\mathbf{h}}_b,\beta_b,\beta_e}=\Pr\!\left(\gamma_e > t_b(\hat{\mathbf{h}}_b)\mid\hat{\mathbf{h}}_b,\beta_b,\beta_e\right).
\end{align}
Equivalently, in terms of the conditional CDF of $\gamma_e$, we have
\begin{align}\label{eq:Pso_condensed}
P_{\mathrm{so}\mid\hat{\mathbf{h}}_b,\beta_b,\beta_e}=\begin{cases}
1 - F_{\gamma_e|\hat{\mathbf{h}}_b,\beta_b,\beta_e}(t_b(\hat{\mathbf{h}}_b)), & t_b(\hat{\mathbf{h}}_b) \ge 0,\\
1, & t_b(\hat{\mathbf{h}}_b) < 0,
\end{cases}
\end{align}
where the second line results from $\gamma_e\ge 0$. Finally, averaging over the joint statistics of $\hat{\mathbf{h}}_b$, $\beta_b$, and $\beta_e$ yields
\begin{align}\label{eq:sop_final_clean}
P_{\mathrm{so}}=\mathbb{E}_{\beta_b,\beta_e}\left[\mathbb{E}_{\hat{\mathbf{h}}_b \mid \beta_b}\left[P_{\mathrm{so}\mid\hat{\mathbf{h}}_b,\beta_b,\beta_e}\right]\right].
\end{align}

If the eavesdropper large-scale gain is only known to lie in the uncertainty set $\mathcal B_e=[\beta_e^{\min},\beta_e^{\max}]$, a robustness-oriented metric is the worst-case SOP given by
\begin{equation}\label{eq:sop_worst_case_final}
P_{\mathrm{so}}^{\mathrm{wc}}\triangleq\sup_{\beta_e\in\mathcal B_e}P_{\mathrm{so}}(\beta_e).
\end{equation}
It should be noted that since the eavesdropper SINR $\gamma_e$ is monotonically increasing with respect to the large-scale gain $\beta_e$, the SOP is also monotonically increasing in $\beta_e$. Therefore, the worst-case SOP occurs at $\beta_e=\beta_e^{\max}$.

\vspace{-2mm}
\section{ESR under Correlated Leakage and AN}\label{sec:ergodic_sec}

\subsection{Definition and Joint-Distribution Representation}
From \eqref{eq:rs}, the ESR for our model is defined as
\begin{align}\label{eq:Rsbar_general}
\bar R_s=\eta\mathbb{E}\left[\left(\log_2(1+\gamma_b)-\log_2(1+\gamma_e)\right)^+\right].
\end{align}
Due to the correlated SW leakage, $\gamma_b$ and $\gamma_e$ are statistically coupled whenever $\rho>0$. Consequently, the expectation in \eqref{eq:Rsbar_general} must be evaluated with respect to their joint distribution rather than separate marginals. Hence, the ESR is found as
\begin{equation}\label{eq:Rsbar_joint}
\bar R_s=\eta\int_0^\infty\int_0^\infty\left[\log_2\!\left(\frac{1+t_b}{1+t_e}\right)\right]^+f_{\gamma_b,\gamma_e}(t_b,t_e)dt_e dt_b,
\end{equation}
where $f_{\gamma_b,\gamma_e}(t_b,t_e)$ is the joint PDF of $(\gamma_b,\gamma_e)$. If $\rho=0$, the joint density factorizes as
$
f_{\gamma_b,\gamma_e}(t_b,t_e)
=
f_{\gamma_b}(t_b) f_{\gamma_e}(t_e)
$. 
However, for $\rho>0$, such factorization does not hold, and the joint statistics are governed by the leakage-induced coupling through the shared channel components.

\vspace{-2mm}
\subsection{Conditional Expectation Decomposition}
To obtain a tractable representation, we condition on $\hat{\mathbf h}_b$. As established earlier, the effective SINR $\gamma_b$ becomes deterministic given $\hat{\mathbf h}_b$, whereas $\gamma_e$ remains random due to the quadratic-form structure of the eavesdropper channel. Using the law of total expectation, the ESR can be written as
\begin{align}\label{eq:Rsbar_conditional}
\bar R_s=\eta\mathbb{E}_{\beta_b,\beta_e}\left[\mathbb{E}_{\hat{\mathbf h}_b \mid \beta_b}\left[\mathcal{G}(\hat{\mathbf h}_b,\beta_e)\right]\right],
\end{align}
where
\begin{multline}
\mathcal{G}(\hat{\mathbf h}_b,\beta_e)=\\
\mathbb{E}\left[\left(\log_2(1+\gamma_b(\hat{\mathbf h}_b))-\log_2(1+\gamma_e)\right)^+\Big|\hat{\mathbf h}_b,\beta_e\right].
\end{multline}
Now, since $\gamma_b$ is deterministic given $\hat{\mathbf{h}}_b$, we have
\begin{align}
\gamma_b^\star\triangleq\gamma_b(\hat{\mathbf h}_b).
\end{align}
Then given that the secrecy is achieved only when $\gamma_e < \gamma_b^\star$, the conditional expectation reduces to
\begin{align}\label{eq:G_definition}
\mathcal{G}(\hat{\mathbf h}_b,\beta_e)=\int_{0}^{\gamma_b^\star}\log_2\!\left(\frac{1+\gamma_b^\star}{1+t}\right)f_{\gamma_e \mid \hat{\mathbf h}_b,\beta_e}(t)dt.
\end{align}
Therefore, the ESR reduces to an outer expectation over the distribution of the channel estimate and the large-scale routing parameters, while the inner integral depends only upon the conditional CDF of $\gamma_e$ derived in Section~\ref{subsec:cond_cdf_ge}.

\vspace{-2mm}
\subsection{Closed-Form Representation under Noncentral Leakage}\label{subsec:ergodic_noncentral}
We proceed to derive a closed-form representation of the conditional expectation $\mathcal{G}(\hat{\mathbf h}_b,\beta_e)$ in the general noncentral case $\kappa \ge 0$. Conditioned on $\hat{\mathbf h}_b$, the eavesdropper SINR admits the CDF in \eqref{eq:conditional_ge_series_fixed}. Now, by differentiating \eqref{eq:conditional_ge_series_fixed} with respect to $t$, the conditional PDF is obtained as
\begin{align}\label{eq:noncentral_pdf}
f_{\gamma_e \mid \hat{\mathbf h}_b}(t)=e^{-\kappa}\sum_{m=0}^{\infty}\frac{\kappa^m}{m!}f_m(t),
\end{align}
where each $f_m(t)$ is the $m$-th Poisson component, given as
\begin{align}\label{eq:fm_component}
f_m(t)=\sum_{i=1}^{M-1}\alpha_i\frac{\exp\!\left(-\delta(t)\right)}{\left(1+\xi(t)\lambda_i\right)^{m+2}}\cdot \frac{\sigma^2 + P_a \lambda_i}{P_s \sigma_X^2}.
\end{align}
Now, substituting \eqref{eq:noncentral_pdf} into
\eqref{eq:G_definition} gives
\begin{multline}\label{eq:G_noncentral}
\mathcal{G}=\frac{e^{-\kappa}}{\ln 2}\sum_{m=0}^{\infty}\frac{\kappa^m}{m!}\sum_{i=1}^{M-1}\alpha_i\\
\times\int_0^{\gamma_b^\star}\frac{\ln\!\left(\frac{1+\gamma_b^\star}{1+t}\right)\exp\!\left(-\delta(t)\right)}{\left(1+\xi(t)\lambda_i\right)^{m+2}}\cdot\frac{\sigma^2 + P_a \lambda_i}{P_s \sigma_X^2}dt.
\end{multline}

The inner integral in \eqref{eq:G_noncentral} admits a closed-form representation in terms of generalized exponential integral functions and incomplete logarithmic integrals. Moreover, the outer Poisson-weighted series converges rapidly due to the factorial decay in $m$, enabling accurate numerical evaluation with only a small number of terms. Additionally, Fig. \ref{fig:pdf_t} demonstrates the accuracy of the analytical results for the derived conditional PDF of the eavesdropper SINR given in \eqref{eq:noncentral_pdf}.

\begin{corollary}[Conditionally Central Case]\label{cor:central_case}
If the leakage-induced conditional mean vanishes, i.e., $\mu_X = 0$ (equivalently, $\kappa = 0$), the Poisson mixture in \eqref{eq:noncentral_pdf} collapses to its zeroth-order term. In this case, the conditional PDF reduces to
\begin{align}
f_{\gamma_e \mid \hat{\mathbf h}_b}(t)=\sum_{i=1}^{M-1}\alpha_i\frac{\exp\!\left(-\delta(t)\right)}{\left(1+\xi(t)\lambda_i\right)^2}\cdot\frac{\sigma^2 + P_a \lambda_i}{P_s \sigma_X^2},
\end{align}
which yields a simplified single-sum integral for $\mathcal{G}(\hat{\mathbf h}_b,\beta_e)$.
\end{corollary}

\begin{figure}[]
\centering
\includegraphics[width=.95\columnwidth]{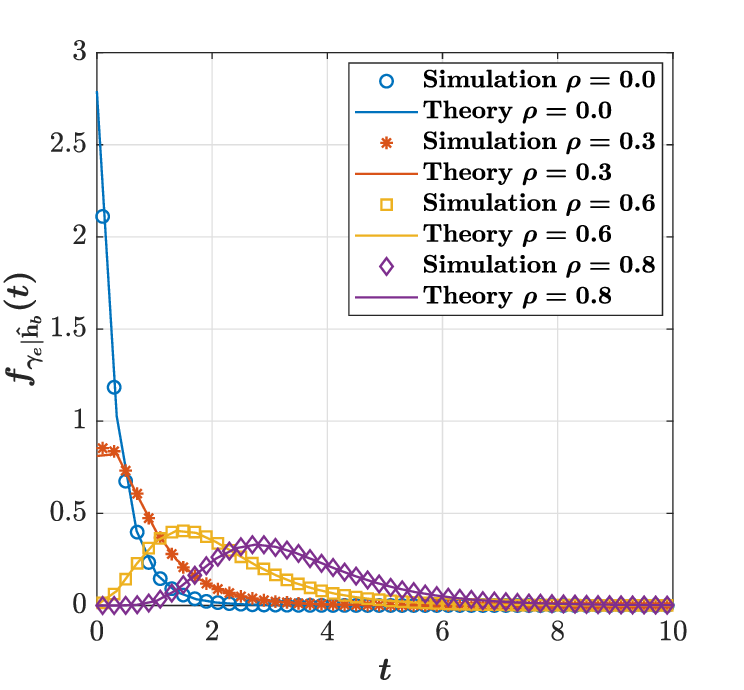}
\caption{Conditional PDF of the eavesdropper SINR under different channel correlation coefficients $\rho$.}\label{fig:pdf_t}
\vspace{-2mm}
\end{figure}

\vspace{-2mm}
\subsection{High-SNR Behavior with and without AN}
We now characterize the asymptotic regime $P\rightarrow\infty$ under a fixed power-splitting policy of $P_s = (1-\theta_a)P$ and $P_a = \theta_a P$, 
where $0 \le \theta_a < 1$ denotes the AN power fraction.

\subsubsection{Pure MRT ($\theta_a=0$)}
When AN is absent, the eavesdropper SINR in \eqref{eq:ge_XY_def} reduces to
\begin{align}
\gamma_e=\frac{P |X|^2}{\sigma^2},
\end{align}
which grows linearly with $P$ almost surely. In contrast, from \eqref{eq:gamma_b_general}, the legitimate SINR converges to a finite random ceiling
\begin{align}
\gamma_b\xrightarrow{a.s.}\frac{\|\hat{\mathbf h}_b\|^2}{\mathbf w^H\mathbf R_{\tilde h}\mathbf w},
\end{align}
due to the residual estimation error. Consequently, $	\bar R_s \rightarrow 0$, and $P\rightarrow\infty$, recovering the well-known secrecy degradation at high SNR under imperfect CSI.

\subsubsection{AN Present ($\theta_a>0$)}
When a fixed non-zero AN fraction is maintained, both the numerator and denominator of $\gamma_e$ scale linearly with $P$. Specifically,
\begin{align}
\gamma_e=\frac{(1-\theta_a)P |X|^2}{\theta_a P Y + \sigma^2}\xrightarrow{a.s.}\frac{(1-\theta_a)|X|^2}{\theta_a Y},
\end{align}
which is finite almost surely. Note that $|X|^2$ remains noncentral when $\kappa>0$, so leakage effects persist in the limit. Meanwhile, $\gamma_b$ converges to
\begin{equation}
\gamma_b\xrightarrow{a.s.}\frac{(1-\theta_a)\|\hat{\mathbf h}_b\|^2}{(1-\theta_a)\mathbf w^H\mathbf R_{\tilde h}\mathbf w +\theta_a\,\mathrm{tr}(\mathbf V^H\mathbf R_{\tilde h}\mathbf V)}.
\end{equation}
Therefore,
\begin{align}
\lim_{P\to\infty} \bar R_s=\eta\mathbb{E}\left[\left(\log_2(1+\gamma_b^\infty)-\log_2(1+\gamma_e^\infty)\right)^+\right],
\end{align}
where $\gamma_b^\infty$ and $\gamma_e^\infty$ are the limiting random variables.

\begin{proposition}[Elimination of Secrecy Collapse]
Under correlated SW leakage and imperfect CSI, secrecy collapse at high SNR occurs if and only if $\theta_a = 0$. For any fixed $\theta_a>0$, the ESR converges to a finite limit as $P\rightarrow\infty$.
\end{proposition}

It is noteworthy that the asymptotic secrecy rate is strictly positive whenever the AN fraction $\theta_a$ is sufficiently large relative to the leakage strength and routing statistics. Hence, secrecy degradation at high transmit power is not intrinsic to the E-FAS architecture, but rather a consequence of beamforming without AN under imperfect CSI.

\vspace{-2mm}
\subsection{Routing-CSI Estimation Coupling in E-FAS and Secrecy-Oriented Power Allocation}\label{subsec:routing_power_split}
This subsection provides two design consequences that are specific to E-FAS: (i) the routing gain $\beta_b$ improves secrecy not only by strengthening the main channel but also by improving CSI quality through $\Omega_{\hat h}$ and $\Omega_{\tilde h}$, and (ii) secrecy at high SNR admits a non-trivial trade-off between data power and AN power, which yields an interior optimal power split.

\subsubsection{Impact of Routing Gain on CSI Quality and Bob's SINR Ceiling}
Recall that the MMSE estimation variances satisfy \eqref{eq:om1} and \eqref{eq:om2}. Therefore, $\Omega_{\tilde h}$ decreases monotonically with $\beta_b$, while $\Omega_{\hat h}$ increases monotonically with $\beta_b$. This coupling is a structural consequence of E-FAS routing because the same routing gain governs both the received pilot quality and the downlink signal strength.

Under $\mathbf{R}_{\tilde h}=\Omega_{\tilde h}\mathbf{I}_M$, Bob's SINR in \eqref{eq:gamma_b_general} reduces to
\begin{align}\label{eq:gamma_b_iso}
\gamma_b=\frac{P_s\|\hat{\mathbf{h}}_b\|^2}{\Omega_{\tilde h}\big(P_s+P_a(M-1)\big)+\sigma^2}.
\end{align}
At high SNR with a fixed AN fraction, i.e., $P_s=\alpha P$, $P_a=(1-\alpha)P$ for $\alpha\in(0,1)$, we obtain the finite ceiling
\begin{align}\label{eq:gamma_b_ceiling_alpha}
\gamma_b^{\infty}=\lim_{P\to\infty}\gamma_b=\frac{\alpha\,\|\hat{\mathbf{h}}_b\|^2}{\Omega_{\tilde h}\big(\alpha+(1-\alpha)(M-1)\big)}.
\end{align}
Since $\hat{\mathbf{h}}_b\sim\mathcal{CN}(\mathbf{0},\Omega_{\hat h}\mathbf{I}_M)$, the random ceiling level is governed by the ratio $\Omega_{\hat h}/\Omega_{\tilde h}$, which increases with $\beta_b$.

\begin{corollary}[Routing Increases Bob's Effective SINR Ceiling]\label{cor:routing_ceiling}
For $\mathbf{R}_{\tilde h}=\Omega_{\tilde h}\mathbf{I}_M$ and any fixed $\alpha\in(0,1)$, the ceiling $\gamma_b^{\infty}$ in \eqref{eq:gamma_b_ceiling_alpha} is stochastically increasing in $\beta_b$. In particular, the mean ceiling satisfies
\begin{align}
\mathbb{E}\{\gamma_b^{\infty}\}&=\frac{\alpha\,M\,\Omega_{\hat h}}{\Omega_{\tilde h}\big(\alpha+(1-\alpha)(M-1)\big)}\notag\\
&=\frac{\alpha\,M\,\tau_p\rho_p\,\beta_b}{\sigma^2\big(\alpha+(1-\alpha)(M-1)\big)},\label{eq:mean_ceiling_closed}
\end{align}
which increases linearly with $\beta_b$.
\end{corollary}

\begin{proof}
The first line follows since $\|\hat{\mathbf{h}}_b\|^2$ is Gamma distributed with scale $\Omega_{\hat h}$ and $\Omega_{\hat h}/\Omega_{\tilde h}=(\tau_p\rho_p\beta_b)/\sigma^2$ is strictly increasing in $\beta_b$ by \eqref{eq:om1} and \eqref{eq:om2}. The mean expression follows from $\mathbb{E}\{\|\hat{\mathbf{h}}_b\|^2\}=M\Omega_{\hat h}$ and simplification of $\Omega_{\hat h}/\Omega_{\tilde h}$.
\end{proof}

\subsubsection{Secrecy-Optimal AN/Data Power Split (Existence of an Interior Optimum)}
We next formalize the data-AN power allocation trade-off. Let $P_s=\alpha P$ and $P_a=(1-\alpha)P$ with $\alpha\in[0,1]$. Increasing $\alpha$ improves Bob's desired signal power but weakens the AN that limits Eve. Conversely, decreasing $\alpha$ strengthens AN at Eve but reduces the useful signal power at Bob and also increases AN-induced sensitivity to estimation error through \eqref{eq:gamma_b_iso}. This yields an interior optimal split.

To obtain a tractable analytical statement, we focus on the high-SNR regime with isotropic estimation error, i.e., $\mathbf{R}_{\tilde h}=\Omega_{\tilde h}\mathbf{I}_M$, and adopt the standard approximation that the ESR is dominated by the difference between logarithms of the limiting SINRs, which is given by
\begin{align}\label{eq:Rs_infty_alpha_def}
\bar R_s^{\infty}(\alpha)\triangleq\eta\,\mathbb{E}\!\left[\left(\log_2(1+\gamma_b^{\infty})-\log_2(1+\gamma_e^{\infty})\right)^+\right],
\end{align}
where $\gamma_b^{\infty}$ is given in \eqref{eq:gamma_b_ceiling_alpha} and
\begin{align}\label{eq:gamma_e_infty_alpha}
\gamma_e^{\infty}=\lim_{P\to\infty}\gamma_e=\frac{\alpha\,|X|^2}{(1-\alpha)\,Y},
\end{align}
with $X=\mathbf{h}_e^H\mathbf{w}$ and $Y=\|\mathbf{h}_e^H\mathbf{V}\|^2$ as in Section~\ref{sec:sinr_analysis_sec}. Note that since $\mathbf V$ has rank $M-1$ and $\mathbf h_e$ is a continuous complex Gaussian vector, the quadratic form $Y=\|\mathbf h_e^H\mathbf V\|^2$ is strictly positive with probability one for $M\ge2$. Besides,  \eqref{eq:gamma_e_infty_alpha} remains finite almost surely for any $\alpha\in(0,1)$.

\begin{lemma}[Non-Triviality of the Interior Secrecy Region]\label{lem:nontriviality}
Assume $\mathbf{R}_{\tilde h}=\Omega_{\tilde h}\mathbf{I}_M$, $M\ge 2$, and $\rho\in[0,1)$ is fixed. Suppose that the eavesdropper has non-degenerate AN leakage, i.e.,
\begin{align}
\Pr\!\left(Y>0\right)=1,
\end{align}
and Bob's effective high-SNR ceiling is non-degenerate, i.e.,
\begin{align}
\Pr\!\left(\|\hat{\mathbf{h}}_b\|^2>0\right)=1.
\end{align}
Then there exists $\alpha_0\in(0,1)$ such that the high-SNR ESR satisfies
\begin{align}
\bar R_s^{\infty}(\alpha_0)>0.
\end{align}
\end{lemma}

\begin{proof}
Fix any $\alpha\in(0,1)$, under AN, we have
\begin{equation}
\gamma_e^{\infty}(\alpha)=\frac{\alpha |X|^2}{(1-\alpha)Y},
\end{equation}
which is finite with probability one because $Y>0$ with probability one. Similarly, from \eqref{eq:gamma_b_ceiling_alpha},
\begin{equation}
\gamma_b^{\infty}(\alpha)=\frac{\alpha\|\hat{\mathbf h}_b\|^2}{\Omega_{\tilde h}\big(\alpha+(1-\alpha)(M-1)\big)},
\end{equation}
which is strictly positive with probability one since $\|\hat{\mathbf h}_b\|^2>0$ with probability one and $\Omega_{\tilde h}>0$.

As $\gamma_e^{\infty}(\alpha)$ is finite with unbounded support in $|X|^2$, for any $\epsilon\in(0,1)$, there exists $t_\epsilon>0$ such that $\Pr\!\left(\gamma_e^{\infty}(\alpha)\le t_\epsilon\right)\ge \epsilon$. Similarly, since $\|\hat{\mathbf h}_b\|^2$ is Gamma distributed with unbounded support, there exists $s_\epsilon>t_\epsilon$ such that $\Pr\!\left(\gamma_b^{\infty}(\alpha)\ge s_\epsilon\right)\ge \epsilon$. 

Now we define the set of channel realizations as
\begin{align}
\mathcal E\triangleq\{\gamma_b^{\infty}(\alpha)\ge s_\epsilon\}\cap\{\gamma_e^{\infty}(\alpha)\le t_\epsilon\}.
\end{align}
Using the union bound, we have
\begin{align}
\Pr(\mathcal E)\ge1-(1-\epsilon)-(1-\epsilon)=2\epsilon-1.
\end{align}
Choosing $\epsilon>1/2$ yields $\Pr(\mathcal E)>0$. On $\mathcal E$, the secrecy integrand is strictly positive, and hence $\bar R_s^{\infty}(\alpha)>0$.
\end{proof}

\begin{proposition}[Interior Optimal Power Split with Nonzero Optimum]\label{prop:alpha_star_strong}
Assume $\mathbf{R}_{\tilde h}=\Omega_{\tilde h}\mathbf{I}_M$, $M\ge 2$, and $\rho\in[0,1)$ is fixed. Then the mapping $\alpha\mapsto \bar R_s^{\infty}(\alpha)$ admits at least one maximizer $\alpha^\star\in(0,1)$. Moreover, $\max_{\alpha\in[0,1]}\bar R_s^{\infty}(\alpha)>0$.
\end{proposition}

\begin{proof}
Continuity of $\bar R_s^{\infty}(\alpha)$ on $[0,1]$ follows from the continuity of $\gamma_b^{\infty}(\alpha)$ and $\gamma_e^{\infty}(\alpha)$ for $\alpha\in(0,1)$ and bounded convergence on compact subintervals. Since $[0,1]$ is compact, $\bar R_s^{\infty}(\alpha)$ attains a maximum. 
The endpoint behavior is characterized as follows. When $\alpha=0$, no information signal is transmitted and thus $\bar R_s^{\infty}(0)=0$. When $\alpha \to 1$, the AN power fraction vanishes and $\gamma_e^{\infty}(\alpha)=\frac{\alpha |X|^2}{(1-\alpha)Y}$ diverges almost surely because $Y>0$ with probability one. Consequently, $\log_2(1+\gamma_e^{\infty}(\alpha))\to\infty$ while $\gamma_b^{\infty}(\alpha)$ remains finite, which implies $\lim_{\alpha\to1}\bar R_s^{\infty}(\alpha)=0$. By Lemma~\ref{lem:nontriviality}, there exists $\alpha_0\in(0,1)$ such that $\bar R_s^{\infty}(\alpha_0)>0$. Therefore, the maximum value is strictly positive and cannot occur only at the endpoints; hence at least one maximizer lies in the interior $(0,1)$.
\end{proof}

\begin{remark}[E-FAS-Specific Interpretation]
Proposition~\ref{prop:alpha_star_strong} is not merely a generic AN statement. In E-FAS, $\beta_b$ enters $\bar R_s^{\infty}(\alpha)$ through the ratio $\Omega_{\hat h}/\Omega_{\tilde h}=(\tau_p\rho_p\beta_b)/\sigma^2$, which controls the ceiling in \eqref{eq:mean_ceiling_closed}. Therefore, stronger E-FAS routing shifts the secrecy-optimal split toward larger $\alpha$ (more data power) while maintaining a strictly positive AN fraction, because improved CSI quality reduces the AN penalty at Bob.
\end{remark}

\subsubsection{A Practical Closed-Form Proxy for $\alpha^\star$}
To have a simple design rule for the secrecy-oriented power split, we construct a conservative proxy by matching the mean high-SNR ceilings of Bob and Eve, following the average-SINR approximations for AN power allocation design \cite{zhou2010secure}.

Recall that the mean high-SNR ceiling of Bob is \eqref{eq:mean_ceiling_closed}. For Eve, we define the effective mean high-SNR ratio as
\begin{equation}\label{eq:eve_proxy_clean}
\bar\gamma_e^{\infty}(\alpha)=\frac{\alpha}{1-\alpha}\frac{\mathbb{E}\{|X|^2\}}{\mathbb{E}\{Y\}},
\end{equation}
where $X=\mathbf{h}_e^H\mathbf{w}$ and $Y=\|\mathbf{h}_e^H\mathbf{V}\|^2$. The quantities $\mathbb{E}\{|X|^2\}$ and $\mathbb{E}\{Y\}$ follow from the second-order statistics of $\mathbf{h}_e$ and depend explicitly upon the leakage parameter $\rho$ and the correlation matrix $\mathbf{R}_e$. Equating \eqref{eq:mean_ceiling_closed} and \eqref{eq:eve_proxy_clean} and solving for $\alpha$ yields the explicit proxy
\begin{equation}\label{eq:alpha_proxy_exact}
\alpha_{\mathrm{proxy}}=\frac{M\tau_p\rho_p\beta_b-(M-1)\sigma^2\frac{\mathbb{E}\{|X|^2\}}{\mathbb{E}\{Y\}}}{M\tau_p\rho_p\beta_b-(M-2)\sigma^2\frac{\mathbb{E}\{|X|^2\}}{\mathbb{E}\{Y\}}},
\end{equation}
which provides a closed-form initialization that captures: (i) the routing-estimation coupling via $\beta_b$, (ii) the leakage strength via $\mathbb{E}\{|X|^2\}/\mathbb{E}\{Y\}$, and (iii) the antenna dimension $M$.

For further simplification, one may approximate $\alpha+(1-\alpha)(M-1)\approx (M-1)$ in \eqref{eq:mean_ceiling_closed}, which yields
\begin{equation}\label{eq:alpha_proxy_approx}
\alpha_{\mathrm{proxy}}^{\mathrm{(approx)}}=1-\frac{\sigma^2 (M-1)}{M\tau_p\rho_p\beta_b}\frac{\mathbb{E}\{|X|^2\}}{\mathbb{E}\{Y\}}.
\end{equation}

In practice, \eqref{eq:alpha_proxy_exact} is used to initialize numerical optimization of $\alpha$. When clipped to $(0,1)$, it provides an accurate
and computationally efficient starting point.

\section{Numerical Results}\label{sec:num}
This section presents numerical results to evaluate the secrecy performance of the proposed routing-assisted E-FAS system. The analytical expressions are validated through Monte-Carlo simulations. Unless otherwise specified, the number of BS antennas is $M=16$, the channel coherence block length is $T_c=200$, and the pilot length is $\tau_p=20$, resulting in a pre-log factor of $\eta=1-\tau_p/T_c$. Also, the pilot transmit power is $\rho_p=10$, and the noise variance is set to $\sigma^2=1$. Transmit power values are reported in dB scale, i.e., $10\log_{10}(P)$. The large-scale fading coefficients are chosen as $\beta_b=5$ for the legitimate link in the E-FAS case, $\beta_b=1$ for the no-E-FAS benchmark (no routing gain), and $\beta_e=3$ for the eavesdropper link. The secrecy rate threshold is $R_{\mathrm{th}}=1$ bps/Hz, and the AN power allocation is controlled by $\alpha=P_s/P$, where $\theta_a=1-\alpha$ represents the AN power fraction.

Consider two leakage scenarios: (i) a correlated case with leakage correlation coefficient $\rho=0.6$, reflecting the routing-induced statistical dependence between the legitimate and eavesdropper channels, and (ii) an independent case with $\rho=0$, which serves as a benchmark. The `No E-FAS' case is included as a baseline to highlight the performance gain achieved by the proposed routing-assisted design. For the ESR analysis, the analytical expression derived in \eqref{eq:G_noncentral} involves an integral with respect to the conditional CDF of the eavesdropper SINR. This integral is evaluated numerically using a Gauss-Laguerre quadrature technique with $20$ nodes, which efficiently computes expectations over Gamma-distributed variables arising from the AN subspace. This approach provides highly accurate numerical evaluation with a small number of quadrature nodes while preserving computational efficiency. All simulation curves are obtained via Monte-Carlo averaging and are shown to match the theoretical results.

Fig.~\ref{fig:sop_p} depicts the SOP versus the transmit power $P$ for the E-FAS-assisted system and the conventional No-FAS baseline under both correlated and independent leakage conditions. As can be observed, the analytical results closely match the simulation curves across the entire power range, confirming the accuracy of the derived SOP characterization. A clear performance advantage of the E-FAS architecture over the No-FAS benchmark can be observed. For a given transmit power, E-FAS consistently achieves a substantially lower SOP. This improvement is mainly due to the routing gain $\beta_b$ introduced by the SW-assisted propagation, which strengthens the effective channel between Alice and Bob. Importantly, the routing gain also enhances channel estimation quality by increasing the estimation variance $\Omega_{\hat h}$ while reducing the error variance $\Omega_{\tilde h}$. The resulting increase in the effective SINR at Bob enlarges the secrecy margin relative to Eve.

The influence of AN power allocation is also evident. When no AN is employed, i.e., $\theta_a=0$, the SOP initially decreases as the transmit power increases but eventually saturates and increases in the high-SNR regime. This behavior reflects the secrecy collapse phenomenon associated with MRT transmission under imperfect CSI. In this regime, the legitimate SINR converges to a finite ceiling due to the non-vanishing channel estimation error under fixed pilot power, while the eavesdropper SINR grows approximately linearly with the transmit power. As a result, further increasing the transmit power deteriorates the secrecy performance. When a nonzero AN fraction is introduced i.e, $\theta_a>0$, the SOP decreases monotonically with transmit power and approaches very small values in the high-power region. The AN component forces both the numerator and denominator of Eve's SINR to scale with transmit power, preventing the unbounded growth of Eve's SINR and eliminating the secrecy collapse observed in the pure MRT case. Moreover, larger AN fractions shift the SOP curves downward in the considered regime. Increasing $\theta_a$ strengthens the AN interference term in the denominator of $\gamma_e$, which suppresses the high-SINR tail of the eavesdropper distribution. Since the SOP metric is dominated by these rare but harmful events, allocating more power to AN effectively reduces the outage probability despite the additional interference leakage experienced by the legitimate receiver.

We can also observe the impact of leakage correlation in this figure. The correlated scenario leads to higher SOP compared with the independent case. This occurs because the correlation introduces a component of the eavesdropper channel aligned with the legitimate channel direction, which increases the mean of the projection $X=\mathbf{h}_e^H\mathbf{w}$. As a result, the numerator of the eavesdropper SINR becomes noncentral, increasing the probability that the eavesdropper achieves a large SINR.

\begin{figure}[]
\centering
\includegraphics[width=.95\columnwidth]{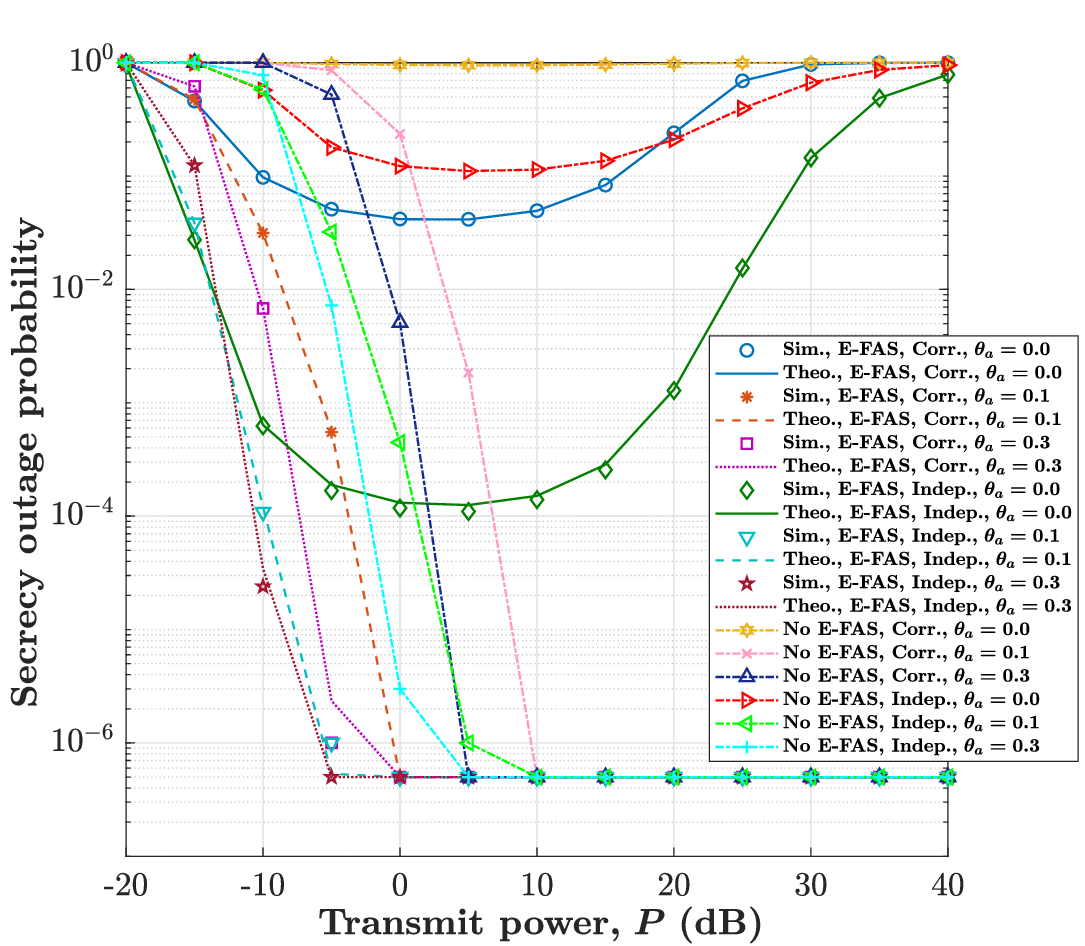}
\caption{SOP versus transmit power $P$ for E-FAS and No-FAS transmission under correlated ($\rho>0$) and independent ($\rho=0$) leakage conditions, with different AN power fractions $\theta_a$.}\label{fig:sop_p}\vspace{-2mm}
\end{figure}
	
Fig.~\ref{fig:esc_p} shows the ESR versus the transmit power $P$ under correlated and independent leakage conditions for different $\theta_s$. The analytical curves are in close agreement with the simulation results, which verifies the numerical evaluation of the integral-form secrecy-rate expression based on the conditional CDF of $\gamma_e$. The ESR increases with the transmit power in the low- and moderate-SNR regimes. In these regimes, the useful signal power at Bob grows faster than the effective information leakage at Eve, so the gap between the two achievable rates widens. The gain is more pronounced in the independent case than in the correlated case, since $\rho=0$ removes the beam-aligned component of the eavesdropper channel and makes the numerator of $\gamma_e$ conditionally central. When $\rho>0$, the projection $X=\mathbf{h}_e^H\mathbf{w}$ becomes noncentral, which increases the average eavesdropper SINR and reduces the secrecy gain.

Furthermore, we can see a qualitatively different behavior which appears when no AN is used. For $\theta_a=0$, the ESR first increases and then decreases at high transmit power. This behavior is a direct result of the imperfect CSI ceiling at Bob together with the unbounded growth of Eve's SINR. In particular, Bob's SINR saturates because the residual estimation error produces an irreducible self-interference term, whereas Eve's SINR scales linearly with $P$ in the absence of AN. As a result, the ESR eventually collapses in the high-SNR regime, in agreement with the asymptotic analysis. When $\theta_a>0$, the high-SNR collapse disappears and the ESR converges to a nonzero ceiling. The AN term forces the denominator of $\gamma_e$ to scale with the transmit power, thereby preventing unbounded growth of the eavesdropper SINR. At the same time, however, AN is not uniformly beneficial from an ergodic-rate perspective. Increasing $\theta_a$ reduces the data power $P_s$ and increases AN leakage at Bob through the estimation error term. Consequently, the ESR is governed by a trade-off between suppressing Eve and preserving Bob's effective SINR. 

This trade-off explains why a larger AN fraction does not necessarily yield a larger ESR. In the correlated case, $\theta_a=0.3$ outperforms $\theta_a=0.1$ because Eve is statistically aligned with Bob and therefore more vulnerable to AN suppression; the reduction in information leakage dominates the loss in Bob's signal power. In the independent case, the opposite ordering is observed at high transmit power. Since the eavesdropper channel is not aligned with the legitimate beam, the marginal benefit of stronger AN is smaller, while the reduction in data power and the AN leakage penalty at Bob remain. Hence, $\theta_a=0.1$ becomes preferable to $\theta_a=0.3$ in the independent scenario. Therefore, this figure highlights two distinct design implications. A strictly positive AN fraction is required to avoid secrecy collapse at high SNR, but the secrecy-optimal AN level depends strongly on the leakage correlation. Stronger statistical coupling between Bob and Eve shifts the preferable operating point toward larger AN fractions, whereas weaker coupling favors a more conservative AN allocation.

\begin{figure}[]
\centering
\includegraphics[width=.95\columnwidth]{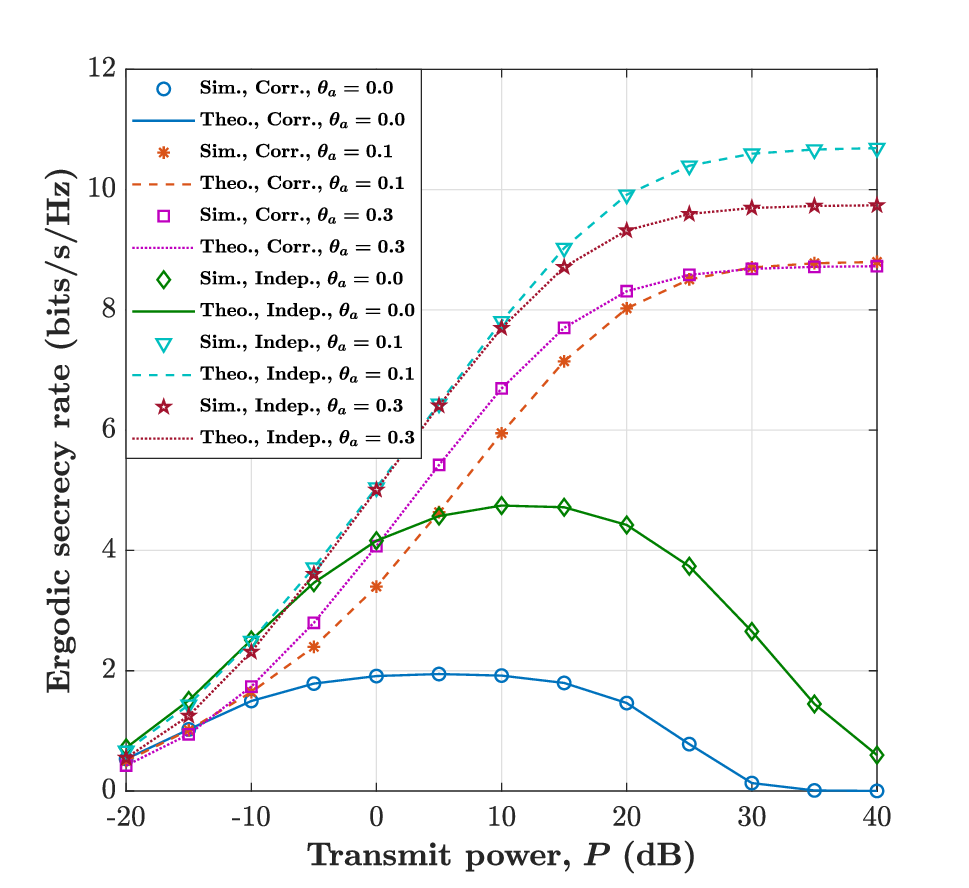}
\caption{ESR versus transmit power $P$ for E-FAS and No-FAS transmission under correlated ($\rho>0$) and independent ($\rho=0$) leakage conditions, with different AN power fractions $\theta_a$.}\label{fig:esc_p}\vspace{-2mm}
\end{figure}

Fig.~\ref{fig:esc_beta} illustrates the ESR versus the large-scale channel gain $\beta_b$ at $P=20$ dB. The ESR increases monotonically with $\beta_b$ for all considered cases. This behavior reflects the routing-estimation coupling introduced by the E-FAS architecture. In particular, increasing $\beta_b$ not only strengthens the effective signal power at the legitimate receiver but also improves the channel estimation quality by increasing the estimation variance $\Omega_{\hat h}$ and reducing the estimation error variance $\Omega_{\tilde h}$. As a result, the effective SINR at Bob increases with the routing gain, which enlarges the secrecy margin relative to the eavesdropper and leads to a higher ESR. Also, a moderate performance reduction can be observed in the presence of leakage correlation, since the eavesdropper channel becomes partially aligned with the legitimate beamforming direction.	

\begin{figure}[]
\centering
\includegraphics[width=.95\columnwidth]{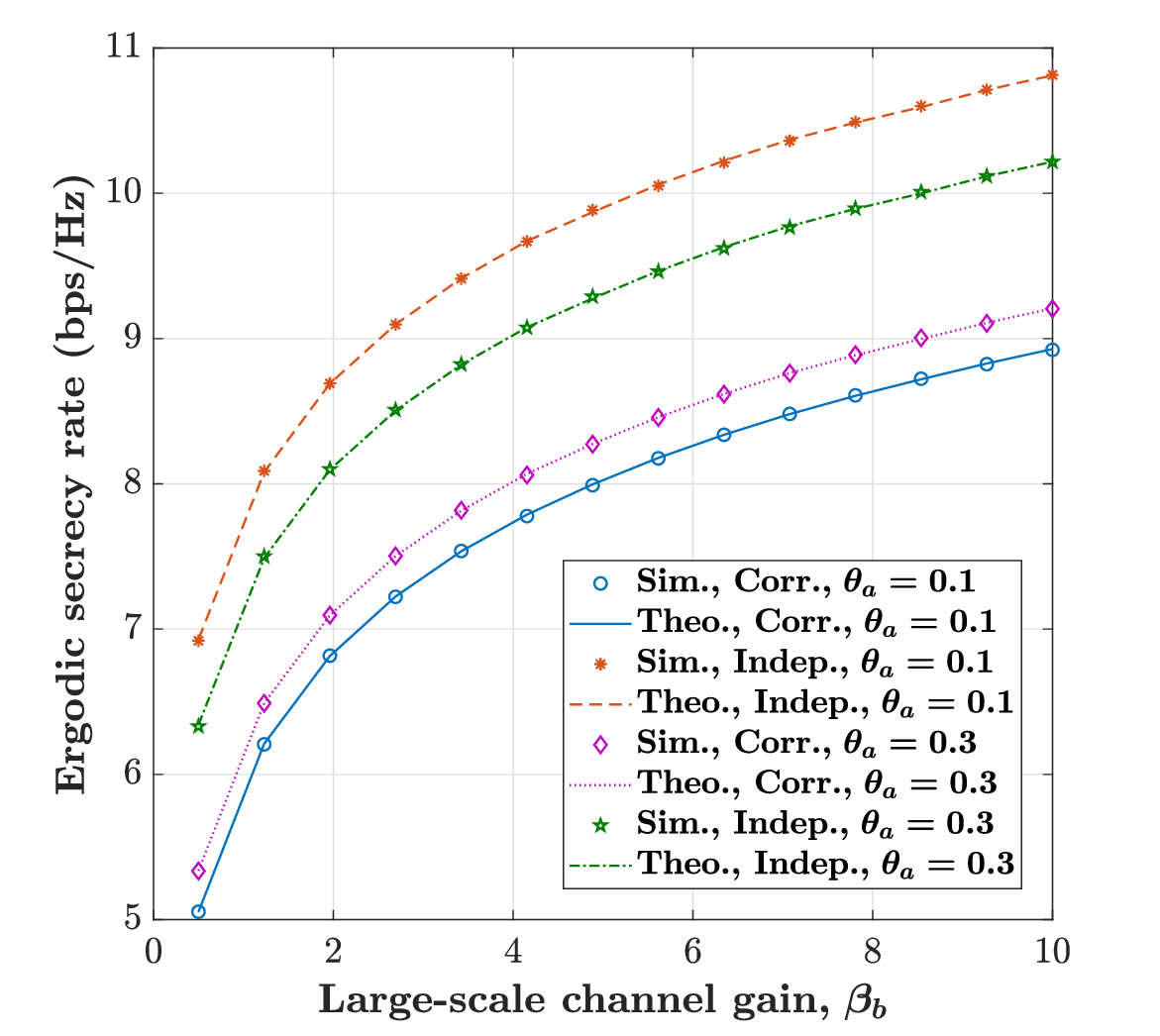}
\caption{ESR versus the large-scale channel gain $\beta_b$ at $P=20$ dB  under correlated ($\rho>0$) and independent ($\rho=0$) leakage conditions, with different AN power fractions $\theta_a$ and fixed $\beta_e=3$. }\label{fig:esc_beta}\vspace{-2mm}
\end{figure}

Fig.~\ref{fig:esc_alpha} represents the ESR versus the data power allocation factor $\alpha=P_s/P$ at $P=20$ dB under correlated and independent leakage conditions. The secrecy rate exhibits a clear unimodal behavior with respect to $\alpha$, indicating the existence of an interior optimal power split between the information signal and AN. When $\alpha$ is small, a large fraction of the transmit power is allocated to AN, which strongly suppresses Eve but also significantly reduces the useful signal power at Bob. As $\alpha$ increases, the signal power at Bob improves and the secrecy rate increases accordingly. However, when $\alpha$ approaches unity, the AN component becomes weak and the protection against eavesdropping diminishes, causing the secrecy rate to decrease. This trade-off leads to an optimal operating point that maximizes the ESR.

This figure also shows that the optimal power allocation depends on the leakage correlation. In the independent case, the optimal value is approximately $\alpha_{\text{indep}}^{\star}\approx0.92$, whereas in the correlated scenario the optimal allocation shifts to a smaller value, $\alpha_{\text{corr}}^{\star}\approx0.76$. This behavior reflects the stronger vulnerability of the system to information leakage when the eavesdropper channel is statistically aligned with the legitimate channel. In such conditions, allocating more power to AN becomes beneficial in order to suppress the increased leakage.

\begin{figure}[]
\centering
\includegraphics[width=.95\columnwidth]{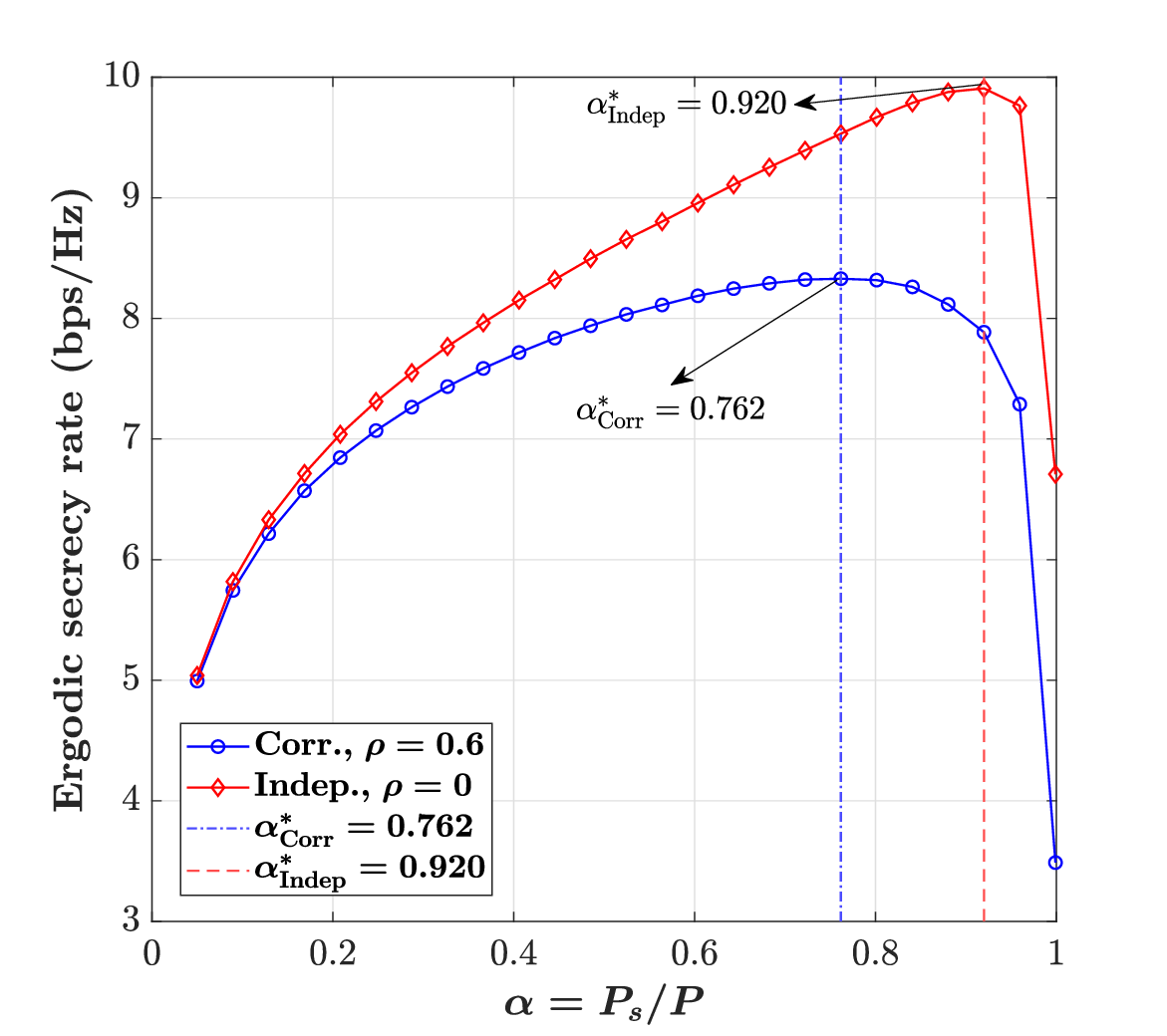}
\caption{ESR versus the data power allocation factor $\alpha=P_s/P$ at $P=20$ dB under correlated ($\rho>0$) and independent ($\rho=0$) leakage conditions.}\label{fig:esc_alpha}\vspace{-2mm}
\end{figure}

The impact of the leakage correlation coefficient $\rho$ on the ESR is illustrated in Fig.~\ref{fig:esc_rho}. As $\rho$ increases, the eavesdropper channel becomes more aligned with the legitimate channel, which strengthens the received information component at the eavesdropper. Consequently, the secrecy rate gradually decreases as the correlation grows. The performance degradation becomes more pronounced in the high-correlation regime, where the eavesdropper is able to exploit the increased statistical dependence between the channels. The figure also highlights the benefit of the E-FAS routing mechanism, where for all the values of $\rho$, the E-FAS configuration consistently achieves higher secrecy rates compared with the conventional No-E-FAS case. This gain originates from the enhanced effective channel gain toward the legitimate receiver, which improves the legitimate link quality while maintaining stronger resilience against information leakage. Moreover, allocating a larger fraction of power to AN provides additional robustness against channel correlation, resulting in improved secrecy performance over the entire range of $\rho$.

\begin{figure}[]
\centering
\includegraphics[width=.95\columnwidth]{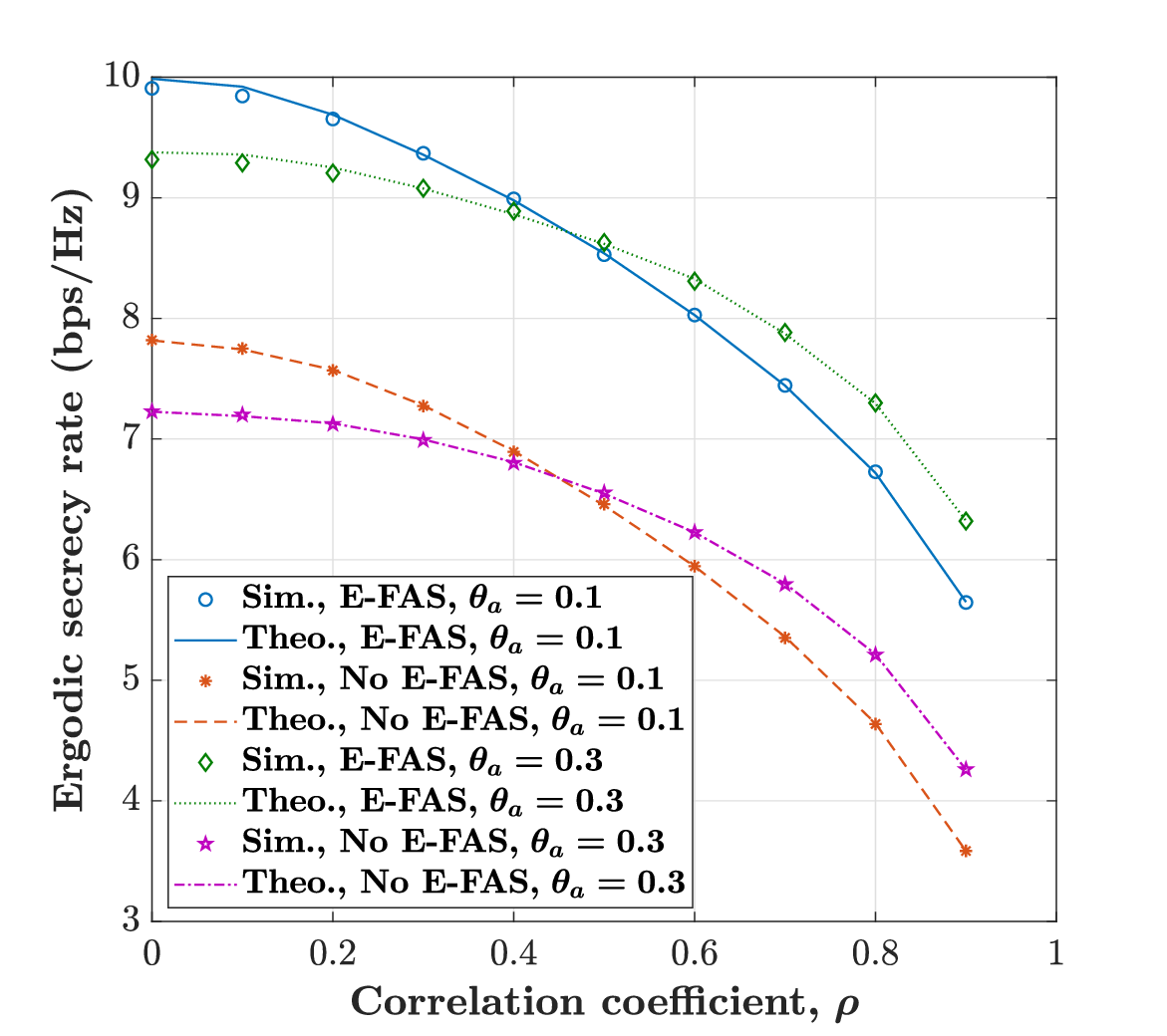}
\caption{ESR versus the leakage correlation coefficient $\rho$ at $P=20$ dB for E-FAS and No-E-FAS configurations under different AN power fractions $\theta_a$.}\label{fig:esc_rho}\vspace{-2mm}
\end{figure}

The ESR as a function of the number of BS antennas $M$ is shown in Fig.~\ref{fig:esc_m}. Increasing the antenna count greatly enhances the secrecy performance, as larger antenna arrays provide higher beamforming gain toward the legitimate receiver while simultaneously improving the spatial suppression of the eavesdropper signal. As a result, the secrecy rate increases rapidly for small values of $M$ and then approaches a saturation level as the antenna count becomes sufficiently large. The figure also shows that the proposed E-FAS consistently achieves higher secrecy rates compared with the conventional No-E-FAS architecture across the entire antenna range. This improvement arises from the enhanced effective channel gain provided by the routing mechanism, which strengthens the legitimate link and improves the overall secrecy performance.

\begin{figure}[]
\centering
\includegraphics[width=.95\columnwidth]{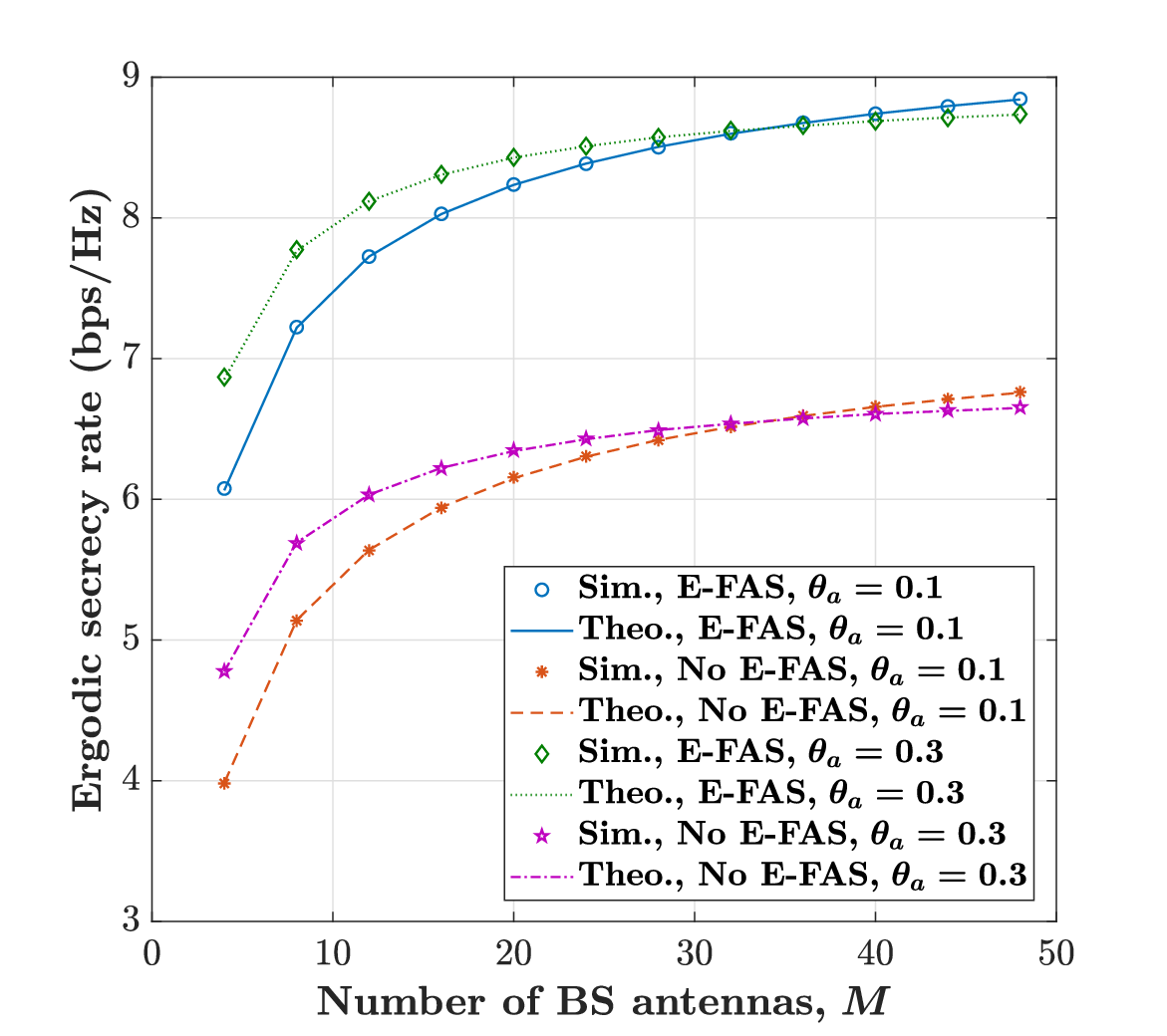}
\caption{ESR versus the number of BS antennas $M$ at $P=20$ dB for E-FAS and No E-FAS configurations under different AN power fractions $\theta_a$.}\label{fig:esc_m}\vspace{-2mm}
\end{figure}

\vspace{-2mm}
\section{Conclusion}\label{sec:con}
This paper investigated the physical layer security performance of routing-assisted E-FAS under practical pilot-based channel estimation. A downlink MISO transmission framework was considered in which the BS performs MMSE channel estimation and applies MRT combined with AN. Unlike conventional independent fading assumptions, a correlated SW leakage model was adopted to capture the statistical coupling between the legitimate and eavesdropper channels induced by electromagnetic routing. By exploiting the two-timescale structure of E-FAS channels, a closed-form conditional characterization of the SOP and a tractable representation of the ESR were derived. The analytical results illustrated several fundamental properties of E-FAS-assisted secure transmission. It was revealed that secrecy collapse at high transmit power occurs only when AN is absent, whereas allocating any strictly positive AN fraction prevents this asymptotic degradation. In addition, the optimal power allocation between the information signal and AN was shown to admit a strictly interior solution, reflecting the inherent trade-off between legitimate link enhancement and eavesdropper suppression. The analysis further indicated that the routing-induced channel gain not only strengthens the desired signal but improves CSI estimation quality, leading to a non-trivial interaction that raises the high-SNR SINR ceiling and disperses secrecy across routing states. Numerical results verified the accuracy of the analytical expressions and illustrated the secrecy advantages by the E-FAS architecture. In particular, the routing gain was shown to significantly expand the secure operating region compared with conventional space-wave transmission while maintaining robustness under correlated leakage conditions. These findings highlight the potential of SW-enabled routing mechanisms as an effective tool for enhancing physical layer security in next-generation of wireless communication systems.\vspace{-3mm}

\appendices
\section{Derivation of \eqref{eq:conditional_ge_series_fixed}}\label{app:cdf48}
Starting from the SINR representation \eqref{eq:ge_XY_def}, where $|X|^2$ is noncentral chi-square with parameter $\kappa$ and $Y$ admits the quadratic-form decomposition $Y=\sum_{i=1}^{M-1}\lambda_i |u_i|^2$, we compute the conditional CDF as
\begin{align}
F_{\gamma_e|\hat{\mathbf h}_b}(t)=\mathbb E_Y\!\left[F_{|X|^2}\!\left(\frac{t(P_a Y+\sigma^2)}{P_s}\right)\right].
\end{align}
Now, using the Poisson-mixture representation of the noncentral chi-square distribution, we have
\[
F_{|X|^2}(x)
=
e^{-\kappa}
\sum_{m=0}^{\infty}
\frac{\kappa^m}{m!}
\left(
1-
e^{-x/\sigma_X^2}
\sum_{k=0}^{m}
\frac{(x/\sigma_X^2)^k}{k!}
\right),
\]
and substituting $x=\frac{t(P_a Y+\sigma^2)}{P_s}$, the expectation over $Y$ introduces its Laplace transform as 
\begin{align}
\mathbb E\left[e^{-sY}\right]=\sum_{i=1}^{M-1}\alpha_i (1+s\lambda_i)^{-1},
\end{align}
valid for distinct eigenvalues $\{\lambda_i\}$. After algebraic manipulation and interchange of summation and expectation, the resulting series expression reduces to \eqref{eq:conditional_ge_series_fixed}.

\bibliographystyle{IEEEtran}

\end{document}